\documentclass[12pt]{article}

\usepackage{graphicx}
\usepackage{epsfig}
\usepackage{amsmath}
\usepackage{amsthm}
\usepackage{amssymb}
\usepackage{algorithm}
\usepackage{algorithmic}
\usepackage{color}
\usepackage[numbers]{natbib}
\usepackage{multirow}

\def\argmin{\mathop{\rm argmin}}
\def\arginf{\mathop{\rm arginf}}

\newcommand{\real}{\ensuremath{\mathbb{R}}}

\newcommand{\ltwo}{\ensuremath{\mathbb{L}^2}}
\newcommand{\lone}{\ensuremath{\mathbb{L}^1}}

\newcommand{\inner}[2]{\left\langle #1,#2 \right\rangle}

\newcommand{\sumton}[2]{{\overset{#2}{\underset{#1}{\sum }}}}

\newtheorem{lemma}{Lemma}
\newtheorem{theorem}{Theorem}
\newtheorem{definition}{Definition}

\begin{document}

\title{\bf Intensity Estimation for Poisson Process with Compositional Noise}
\author{Glenna Schluck, Wei Wu, Anuj Srivastava \\ 
{\it Florida State University} }
\date{}

\maketitle


%
%
%
%
%

\begin{abstract}
Intensity estimation for Poisson processes is a classical problem and has been extensively studied over the past few decades.  Practical observations, however, often contain compositional noise, i.e. a nonlinear shift along the time axis, which makes standard methods not directly applicable.   The key challenge is that these observations are not ``aligned'', and registration procedures are required for successful estimation.  In this paper, we propose an alignment-based framework for positive intensity estimation.  We first show that the intensity function is area-preserved with respect to compositional noise. Such a property implies that the time warping is only encoded in the normalized intensity, or density, function.  Then, we decompose the estimation of the intensity by the product of the estimated total intensity and estimated density.  The estimation of the density relies on a metric which measures the phase difference between two density functions.  An asymptotic study shows that the proposed estimation algorithm provides a consistent estimator for the normalized intensity.  We then extend the framework to estimating non-negative intensity functions.  The success of the proposed estimation algorithms is illustrated using two simulations. Finally, we apply the new framework in a real data set of neural spike trains, and find that the newly estimated intensities provide better classification accuracy than previous methods. 
\end{abstract}

%

%

\section{Introduction}
The study of point processes is one of the central topics in stochastic processes and has been widely used to model discrete events in continuous time.  In particular, the Poisson process, a common point process, has the most applications \cite{Chiang-etal2005, kolaczyk99, nowak_timmermann1998}. Classical examples include the arrivals of park patrons at an amusement park over a period of time, the goals scored in an association football match, and the clicks on a particular web link in a given time period.  Recently, Poisson processes have been used to characterize spiking activity in various neural systems \cite{brown1998, brockwell2004}. In order to use a Poisson process in applications, one key step is to estimate its intensity function from a given sequence of observed events. 

The estimation of the intensity function of a Poisson process has been studied extensively and various estimation methods have been proposed. If the intensity can be assumed to have a known parametric form, then likelihood-based methods can be used to estimate the model parameters. However, in many cases, the shape of the intensity is unknown and estimation requires the implementation of non-parametric methods. Non-parametric estimation methods provide more flexibility than parametric methods and can better characterize the underlying intensity function.  A number of approaches have been proposed over the past three decades, including wavelet-based methods \cite{donoho93, kolaczyk99, rb2010} and kernel-based methods \cite{Bartoszynski-etal1981, Chiang-etal2005, diggle1985}.  In the case where prior knowledge about the process or shape of the intensity is known, Bayesian methods can be adopted and they often lead to a more accurate estimation  \cite{arjas-gasbarra1994, guida-etal1989, timmermann-nowak1997}.  

Treating a neural spike train as a realization of Poisson process, one can consider the example depicted in Fig. \ref{fig:spikexintro}. In this case, the neural spiking activity, which is associated with certain movement behavior \cite{WuSrivastavaJCNS11}, was recorded (see detail of the data in Sec. \ref{subsec:spike1}). The process was repeated for 30 trials and the resulting spike trains are shown in Fig. \ref{fig:spikexintro}A. Notice that in each repetition of the same movement, there is a gap in the spikes that occurs at slightly different times with variable lengths. This time shift in the gap in spikes is indeed an example of the notion of phase variability or compositional noise, a central topic in functional data analysis. The observed gap in spikes should be reflected in the underlying Poisson intensity estimate. However, using kernel-based estimation methods without accounting for phase variability results in an intensity estimate (shown in red in Fig. \ref{fig:spikexintro}B) that does not capture this gap in spiking activity. The method introduced later in this paper does consider the presence of phase variability and yields an estimate of the underlying intensity of the spike train that clearly depicts the observed gap in the spiking activity (shown in blue in Fig. \ref{fig:spikexintro}B). Therefore, it is important to develop estimation procedures that consider the presence of phase variability in repeated observations of the same process, and that is the goal of this paper. 

\begin{figure}[ht]
\begin{center}
\begin{tabular}{cc}
\textbf{A} & \textbf{B} \\
\includegraphics[height=1.7in]{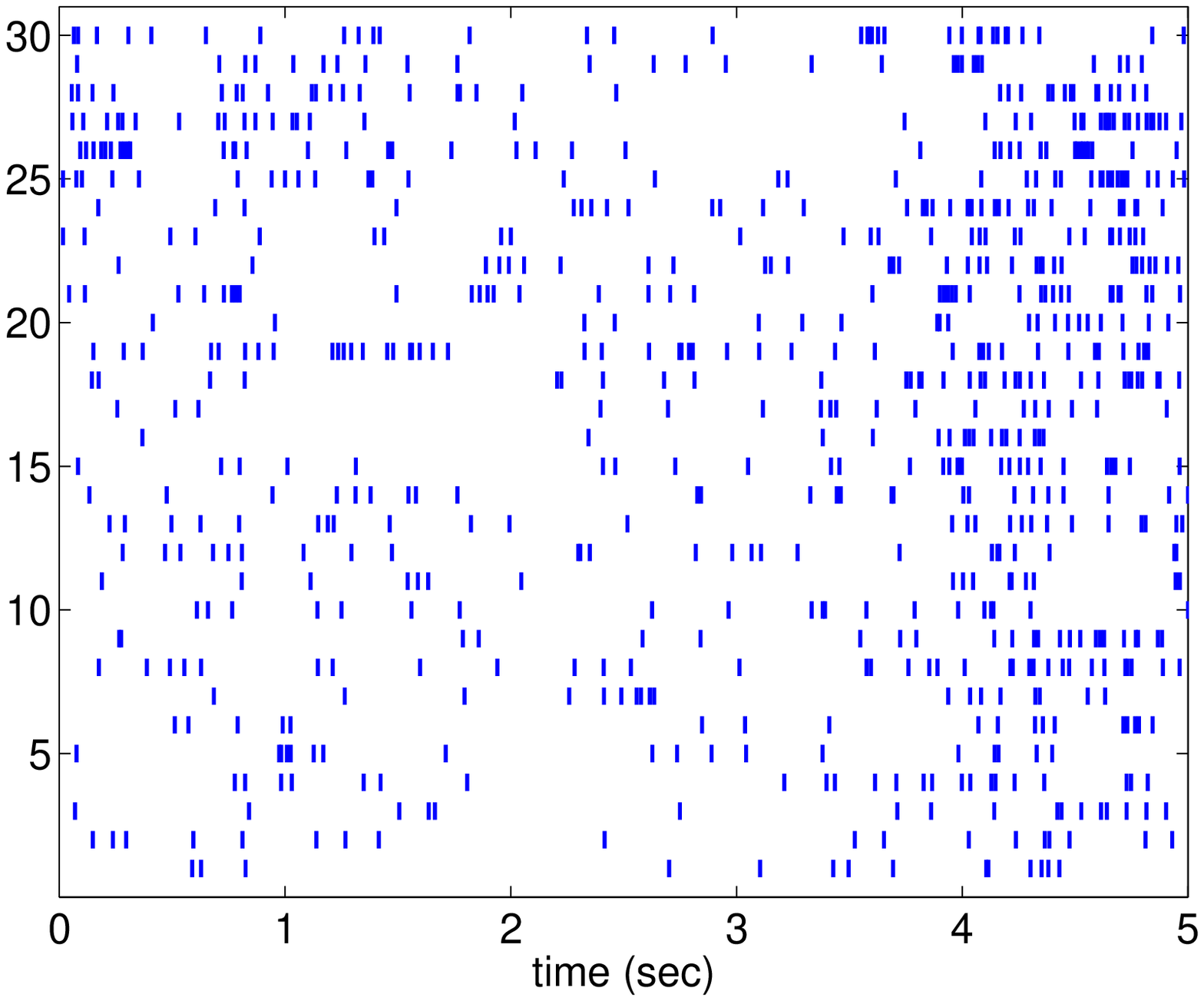}& 
\includegraphics[height=1.7in]{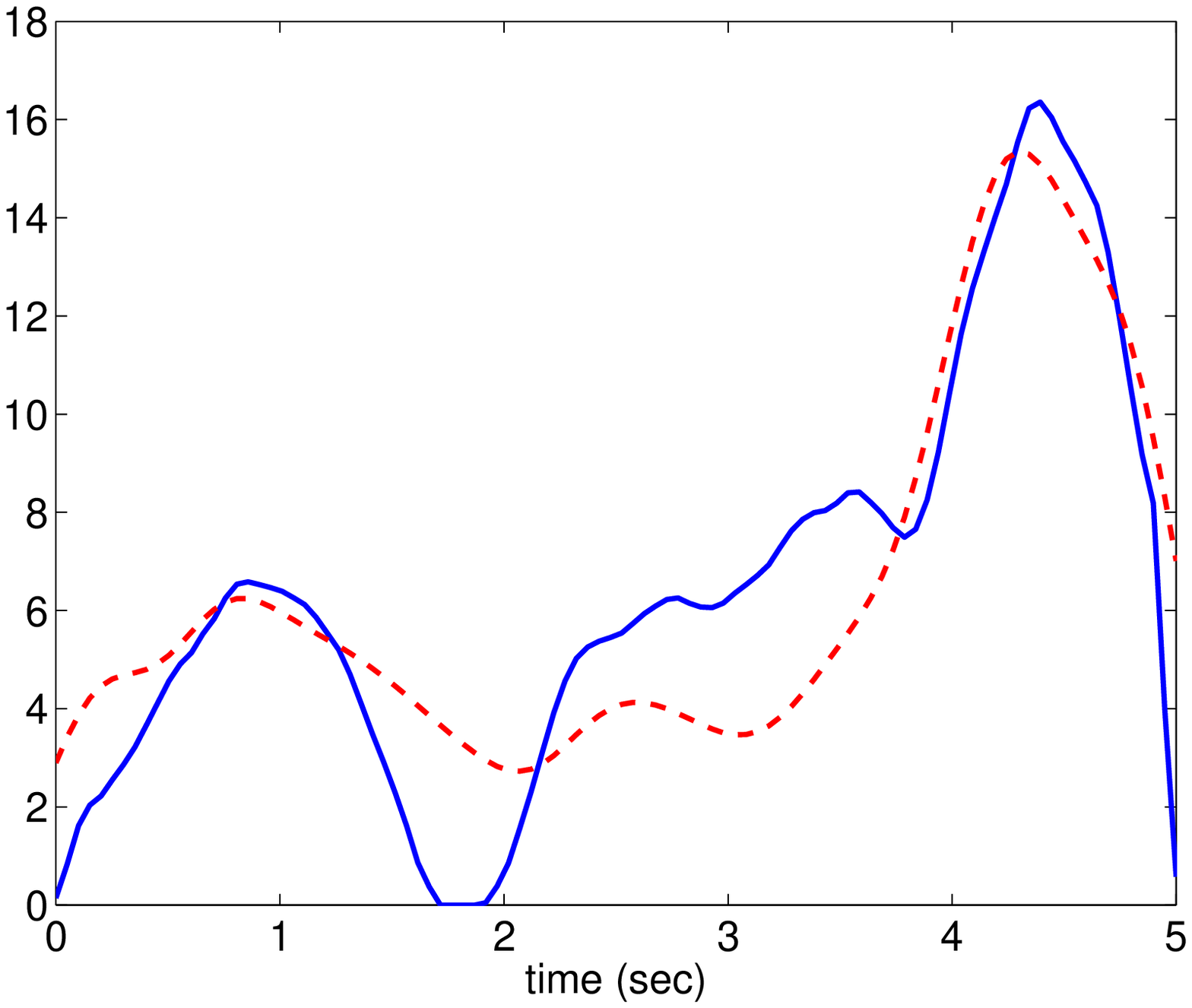}
\end{tabular}
\caption{Intensity estimation example. \textbf{A}. 30 spike trains. \textbf{B}. Estimated intensities by considering (solid blue) and not considering (dashed red) compositional noise.} \label{fig:spikexintro}
\end{center}
\end{figure}

One key concept in functional data analysis where phase variability plays the central role is the notion of function registration, or alignment. Indeed, function registration is an important topic in functional data analysis and a significant amount of research progress has been made over the past two decades \cite{ramsay-li-RSSB:98, gervini-gasser-RSSB:04, muller-biometrika:2008, james:10, TuckerCSDA13}.  
In order to properly register functions, one must consider two types of variability present in data: phase and amplitude variability. Phase variability describes the degree of ``unalignment'' in the data, and amplitude variability is the remaining variability in the vertical axis after alignment. The goal of function registration is to align the functions by removing phase variability. If analysis (such as principal component analysis or regression) is conducted on data which are not well aligned, one may obtain poor or undesired results. 

While function registration has been extensively studied, the notion of aligning point processes with compositional noise has not been well studied -- all aforementioned intensity estimation methods are based on the assumption there is no phase variability in the observed processes. However, as indicated in the above spike train example, that is not always a reasonable assumption.  
To understand the phase variability in point process observations, recent studies on intensity estimation in Poisson process have begun to identify and remove compositional noise during the estimation procedure.  For example,
Bigot and colleagues examined the estimation of the underlying intensity function for a set of linearly shifted Poisson processes \cite{bigot13}. They assumed that the intensity function is periodic and each realization of the process is warped according to a linear shift in time that follows a known distribution. Under the stated assumptions, the authors derived a wavelet-based estimator. They argued that the assumption of a linear shift in the observed processes is a reasonable assumption, particularly in an example of DNA Chip-Seq data. 
However, there are many other cases where it is not reasonable to assume that the phase variability is a simple linear shift such as examples in the literature of functional registration \cite{ramsay-li-RSSB:98, TuckerCSDA13}. As a result, restricting the warping function to be strictly linear shifts may limit the general applicability of their method. In another recent study, Panaretos and Zemel proposed to separate amplitude and phase variation in order to align point processes \cite{panaretos16}. Basically, they extended the notion of the separation of phase and amplitude variation in functions to that of point processes. While the work of Panaretos and Zemel applies generally to point processes, the goal of their work is fundamentally different from the goal of the work in this project. Their goal is estimation of the probability measure and they comment that their work is not to be used for density estimation (see Section 3.4 of \cite{panaretos16}). The goal of the work in this project is intensity estimation, which will be shown reduces to density estimation.

In this paper we propose a new framework for intensity estimation of a Poisson process with compositional noise.  We show that the noise is only encoded in the normalized intensity, or density, function.   The estimation is based on our proposed metric which measures the phase difference between two density functions so the notion of the Karcher mean can be applied in the given framework.  Since the only parameter in the method is the bandwidth for the kernel density estimate, the proposed method is a mostly non-parametric method that yields a consistent estimator of the underlying intensity function.

The rest of this paper is organized as follows. In Section \ref{sec:method}, we present the new framework for positive intensity estimation and discuss its mathematical and computational properties.  Consistency theory on the estimation algorithm is given in Section \ref{sec:theory}.  In Section \ref{sec:nonnegdf}, we extend the estimation to nonnegative intensity functions.  The estimations on positive and nonnegative intensity are illustrated with two simulations, respectively, in Section \ref{sec:results}.  We then show the application of intensity estimation in a real dataset of neural spike trains.  Section \ref{sec:diss} summarizes the work.  Finally, all mathematical details are given in the Appendix.

\section{Method}\label{sec:method}
In this section, we present the new framework for positive intensity estimation of a Poisson process with non-linear time warping. Compositional noise is represented with time warping functions and, since the intensity of a Poisson process is a function, the representation of time warping is studied in the function space. The notation and representation in the function space that is given here is consistent to that in \cite{kurtek-wu-srivastava-NIPS:2011, srivastava-etal-FisherRao:11}. First, we will review the basics of Poisson processes \cite{ross10} and the representation of time warping in function space \cite{kurtek-wu-srivastava-NIPS:2011}.

\subsection{Review of Poisson Process and Time Warping Representation}\label{sec:pp}
  A Poisson process on the time domain $[0, 1]$ is a special type of counting process $N(t), t \in [0,1]$.   For simplification of notation, we only examine the domain $[0,1]$ in this paper, and the framework can be easily adapted to any finite time interval. In the classical theory of point processes, a Poisson process is defined based on an intensity function $\lambda(t) \ge 0$ and satisfies the following two conditions \cite{ross10}:
\begin{enumerate}
\item
Disjoint intervals have counts that are independent. In other words, the number of events occurring in the interval $(a, b)$ is independent of the number of events occurring in the interval $(c, d)$ if these two intervals are not overlapping.
\item 
The number of events in an interval $(a, b) \subset [0,1]$ follows a Poisson distribution with mean $\int_a^b \lambda(t)dt$. In other words, $$P(N(b) - N(a)=n) = \exp\left(-\int_a^b \lambda(t)dt\right)\dfrac{\left(\int_a^b \lambda(t)dt\right)^n}{n!}.$$
\end{enumerate}
We denote a Poisson process with intensity $\lambda(t)$ as $PP(\lambda(t))$.  For distinction, a Poisson distribution with mean $\mu$ is denoted as $Poisson(\mu)$, and a Poisson probability mass function with mean $\mu$ at $k$ is denoted as $Poisson(k; \mu) = e ^{-\mu} \mu^k / k!$.

We represent compositional noise with time warping functions. Since the intensity of a Poisson process is a function, we study the representation of time warping in the function space. 
Let $\Gamma$ be the set of all warping functions, where time warping is defined as an orientation-preserving diffeomorphism of the domain $[0, 1]$. That is, $$\Gamma = \{\gamma: [0,1] \to [0,1] | \gamma(0) =
0,\ \gamma(1)=1,\  0 < \dot \gamma < \infty\}.$$ Elements of $\Gamma$ form a {\em group} with function composition as the group action, and the identity in this group is the self-mapping $\gamma_{id}(t) = t$. For  any function $h$, we will use $\Vert h \Vert$ to denote its $\ltwo$ norm $(\int_0^1 h(t)^2 dt)^{1/2}$.

There are three different types of {\em (right) group actions} about time warping that can occur in the function space:
\begin{enumerate}
\item 
Amplitude-preserved: $f \rightarrow f \circ \gamma$,
\item 
Area ($\mathbb L^1$ norm)-preserved: $f\rightarrow \left(f \circ \gamma\right)\dot{\gamma}:=\left(f; \gamma\right)$,
\item 
Energy ($\mathbb L^2$  norm)-preserved: $f\rightarrow \left(f \circ \gamma\right)\sqrt{\dot{\gamma}}:=\left(f, \gamma\right)$,
\end{enumerate}
where $\circ$ denotes the conventional function composition. 
The properties on associativity and isometry of these three group actions are summarized in Table \ref{tab:props}. In particular, the amplitude-preserved group action is the conventional registration for functions with phase variability and has been extensively studied over the past two decades \cite{ramsay1998, kneipetal2000, kneip-ramsay:2008}.  The enery-preserved group action plays an essential role in the Fisher-Rao registration framework \cite{srivastava-etal-FisherRao:11}, where this action is applied in the Square-Root Velocity Function (SRVF) space (note: it is critical that in the Fisher-Rao framework there is a one-to-one correspondence between the energy-preserved SRVF space and the amplitude-preserved observational function space).  In the following sections of this manuscript, we will show that the compositional noise in the Poisson process intensity function is properly characterized by the area-preserved group action.

\begin{table}[ht]
\begin{center}
\caption{Properties of the three group actions}\label{tab:props}
\begin{tabular}{|c|c|c|c|}
\hline
Group Action & Associativity  & Isometry   \\
\hline
Amplitude-preserved & $ (f \circ \gamma_1) \circ \gamma_2 =  f \circ (\gamma_1 \circ \gamma_2) $ 
			 &  $\| f_1 \circ \gamma - f_2 \circ \gamma \|_{\mathbb{L}^\infty} = \|f_1 - f_2\|_{\mathbb{L}^\infty} $  \\ 
			\hline
Area-preserved  & $ ((f; \gamma_1); \gamma_2) =  (f; (\gamma_1 \circ \gamma_2)) $  &  $\| (f_1; \gamma) - (f_2; \gamma)\|_{\mathbb{L}^1} = \|f_1 - f_2\|_{\mathbb{L}^1} $\\ 
			\hline
Energy-preserved &   $ ((f, \gamma_1), \gamma_2) =  (f, (\gamma_1 \circ \gamma_2) )$   &  $\| (f_1, \gamma) - (f_2, \gamma)\|_{\mathbb{L}^2} = \|f_1 - f_2\|_{\mathbb{L}^2} $ \\ 
\hline
\end{tabular}
\end{center}
\end{table}

\subsection{Poisson Process with Compositional Noise}

Before formally stating the main problem, we review the classical estimation problem in Poisson processes: Given a set of independent realizations from a Poisson process on $[0, 1]$, how can we estimate the underlying intensity function?  By notation, the set of realizations are given in the following form, 
$$ R^i = (r_1^i, \cdots, r_{k_i}^i) \sim PP(\lambda(t)),$$ 
where $k_i \sim Poisson\left(\int_0^1\lambda(t)dt\right), \, i=1,2,\dots,n$. 
Various computational approaches have been developed to address this problem, which include penalized projection estimators \cite{rb2003}, wavelet methods \cite{willett2007, donoho93}, and estimators based upon thresholding rules \cite{rb2010}.   

In this paper, we assume the observed data are not $\{R^i\}$, but a warped version in the form $$S_i = (s_1^i, \cdots, s_{k_i}^i) = \gamma_i^{-1}(R^i) = (\gamma_i^{-1}(r_1^i), \cdots, \gamma_i^{-1}(r_{k_i}^i)),$$ where $\gamma_i$ is a random time warping in $\Gamma$, $i = 1, \cdots, n$.   That is, 
\begin{equation}
S_i = \gamma_i^{-1}(R^i),  \  \mbox{with } \ R^i \sim PP(\lambda(t)), \gamma_i \in \Gamma. 
\label{eq:warping} 
\end{equation}
{\em Given observations $\{S_i\}$, our goal is still to estimate the underlying intensity $\lambda(t)$}. 
To make the model identifiable, we add the constraint that the mean of $\{\gamma_i\}$ needs to be a scaled version of $\gamma_{id}$ (The detail on assumptions is clearly provided in Sec. \ref{sec:theory}).  Note that since the time warping can be in any nonlinear form, this estimation problem is a significant challenge.  
A recent study only examines the case when the warping is a simple linear shift along the time axis \cite{bigot13}.  

As the warping function $\gamma_i$ is random, the warped process $\gamma_i^{-1}(R_i)$  is no longer a Poisson process, but a Cox process.  Here we study,  ``Conditional on $\gamma_i$, is $\gamma_i^{-1}(R^i)$ still a Poisson process?  If this is true, what is the intensity function of that Poisson process?''  Our answer is yes to the first question and the intensity function of the new Poisson process is given as follows.  

\begin{lemma}
\label{lem:int_warp}
Suppose $R$ is a Poisson process with intensity $\lambda(t)$ on $[0,1]$ and $\gamma\in\Gamma$ is a given time warping function. 
Then $\gamma^{-1}(R)$ is also a Poisson process with intensity $\lambda\left(\gamma(t)\right)\dot{\gamma}(t).$ 
\end{lemma}
\begin{proof}
If $R$ is a Poisson process, then the number of events of $R$ in the time interval $(a,b)$ is independent of the number of events of $R$ in the time interval $(c,d$) if $(a,b)\cap (c,d)=\emptyset$. Since $\gamma(t)$ is strictly increasing, \\
$$(a,b)\cap (c,d)=\emptyset\Leftrightarrow (\gamma(a),\gamma(b))\cap (\gamma(c),\gamma(d))=\emptyset. $$ 
Hence, the number of events in $(\gamma(a),\gamma(b))$ is also independent of the number of events in $(\gamma(c),\gamma(d))$. \\

For any $k \in \{0, 1, , \cdots\}$ and sub-interval $[t, t+\Delta t] \subset [0, 1]$, 
\begin{eqnarray*}
&  & P(k \mbox{ events of } \gamma^{-1}(R) \mbox{ are in } [t,t+\Delta t])  \\
&=& P(k \mbox{ events of } R \mbox{ are in } [\gamma(t),\gamma(t+\Delta t)])  \\
&=& Poisson \left(k;\displaystyle\overset{\gamma(t+\Delta t)}{\underset{\gamma(t)}{\int}} \lambda(v)dv\right) \\
&=& Poisson \left(k;\displaystyle\overset{t+\Delta t}{\underset{t}{\int}} \lambda(\gamma(u)) \dot \gamma(u)du\right).\\
\end{eqnarray*}
The last equality holds simply by the change of variable $v = \gamma(u)$.  Therefore, 
$$\gamma^{-1}(R) \sim PP\left(\lambda(\gamma(u))\dot{\gamma}(u)\right).$$

\end{proof}

A direct result from Lemma \ref{lem:int_warp} is that given $\gamma_i$, $S_i$ is also a Poisson process and 
$$ S_i | \gamma_i = \gamma_i^{-1}(R^i) | \gamma_i  \sim PP(\lambda(\gamma_i(t)) \dot \gamma_i(t)). $$  
Based on the theory of Poisson processes, the intensity function $\lambda(t)$ can be decomposed into the product of the total intensity $\Lambda$ and the density function $f(t)$, 
where 
$$ \Lambda=\int_0^1 \lambda(t)dt \ \ \mbox{and}  \ \ f(t) = \lambda(t)/\Lambda.$$  
Therefore, the intensity estimation problem can be reduced to density estimation and scalar total intensity estimation.  

Note that for $i = 1, \cdots, n$, 
$$ \int_0^1 \lambda_i(t)dt = \int_0^1 \lambda(\gamma_i(t))\dot \gamma(t) dt = \int_0^1 \lambda(s)ds = \Lambda.$$
That is,  $\Lambda$ is constant with respect to time warping. 
Hence, the density of the events in $S_i$, given $\gamma_i$, can be written as $f_i(t) = \lambda_i(t) / \Lambda=\lambda(\gamma_i(t))\dot{\gamma}_i(t)/\Lambda$. This expression indicates that the time warping is encoded in the density function, and independent of total intensity. By the theory of Poisson processes, the number of events in each process follows a Poisson distribution with mean $\Lambda$.  For a set of given observations $\{S_i\}$, $\Lambda$ can be easily estimated using a conventional maximum likelihood estimate. Therefore, the intensity estimation problem reduces to estimating the underlying density $f$.  Given $\{S_i\}$, we propose a modified kernel method to estimate density functions $\{f_i\}$, and then use these densities to estimate $f$.  This whole procedure is described in detail in Section \ref{sec:estimation}.

\subsection{Phase Distance Between Positive Probability Density Functions}
\label{sec:distances}

In this paper, we focus on a metric-based method to estimate the underlying density $f$.  Metric distances between density functions is a classical topic and a number of measures have been proposed, for example, the Bhattacharyya Distance \cite{bhattacharyya-43}, the Hellinger Distance \cite{hellinger1909}, the Wasserstein Distance \cite{wasserstein69} and the elastic distance beween densities based upon the Fisher-Rao metric \cite{srivastava-etal-Fisher-Rao-CVPR:2007}. Suppose $f_1$ and $f_2$ are two density functions on [0,1] with cumulative distribution functions $F_1$ and $F_2$, respectively. Then, these metrics are defined as: 
\begin{itemize}
\item Wasserstein Distance: $d_{W}(f_1,f_2)=\Vert F_1^{-1}-F_2^{-1}\Vert$
\item Bhattacharyya Distance: $d_{B}(f_1,f_2)=-\log\left(\int \sqrt{f_1(t)f_2(t)}dt\right)$
\item Hellinger Distance: $d_{H}(f_1,f_2)=\frac{1}{\sqrt{2}}\Vert \sqrt{f_1}-\sqrt{f_2}\Vert$
\item Fisher-Rao Distance: $d_{FR}(f_1,f_2)=\cos^{-1}\left(\int_0^1 \sqrt{f_1(t)}\sqrt{f_2(t)}dt\right)$
\end{itemize}
Note that the Fisher-Rao metric between two density functions is similar to the Hellinger Distance (arc length vs. chord length) \cite{srivastava-etal-Fisher-Rao-CVPR:2007}.  

Based on the generative model in Eqn. \ref{eq:warping}, the difference between the true underlying density function and the noise-contaminated density is the time warping along the time axis. Such a difference is characterized as the {\it phase difference} and we expect that a metric measuring phase difference will be purely based on the warping function between two densities.  That is, the distance between $f_1$ and $f_2$ will only depend on $\gamma$ if $f_1=(f_2; \gamma)$. However, none of the above metrics purely measure this phase difference between two density functions. We aim to find a metric that can properly characterize such phase difference. In this paper, we will define a new distance between positive densities which properly measures their phase difference.  The set of all positive density functions on $[0, 1]$ is denoted as $\mathbb P$.  

We note that for any densities $f_1, f_2 \in \mathbb P$, their cumulative distribution functions $F_1, F_2$ are warping functions in $\Gamma$.  By the group structure of $\Gamma$, it is straightforward to find that the optimal warping function between $f_1$ and $f_2$ (i.e. $\gamma^* \in \Gamma$ such that $f_1=(f_2\circ\gamma^*)\dot{\gamma}^*$ or $F_1=F_2\circ \gamma^*$), is unique and has a closed-form solution given by 
\begin{equation}
\gamma^*=F_2^{-1}\circ F_1. 
\label{eq:optwarp}
\end{equation}
Based on this result, it is natural to define a distance that measures the phase difference by measuring how far the warping function is from the identity warping function, $\gamma_{id}$. In other words, the smaller the distance between the warping function and $\gamma_{id}$, the less warping that is required between the two densities. One definition of the distance metric is given as follows.

\begin{definition}
For any two functions $f_{1}, f_{2} \in \mathbb P$, we define an intrinsic distance, $d_{int}$, between them as: 
\begin{equation}
d_{int}(f_1,f_2)=\arccos \inner{1}{\sqrt{\dot{\gamma}}}\label{eq:disint}
\end{equation}
where $\gamma$ is the optimal time warping between $f_1$ and $f_2$ (i.e. $f_1=(f_2\circ\gamma)\dot{\gamma}$).
\end{definition}\label{def:distint}

This definition of phase distance has been used in the Fisher-Rao framework \cite{TuckerCSDA13}.  This distance is intrinsic which measures the arc-length between ${\sqrt{\dot{\gamma}}}$ and 1 in the unit sphere $S^\infty$ (SRVF space of $\Gamma$).  Note that the definition of phase distance in  $\mathbb P$ is not unique.  We can also define an {\em extrinsic} distance as follows:

\begin{definition}
For any two functions $f_{1}, f_{2} \in \mathbb P$, we define an extrinsic distance, $d_{ext}$, between them as: 
\begin{equation}
d_{ext}(f_1,f_2)=\Vert 1-\sqrt{\dot{\gamma}}\Vert
\label{eq:dis}
\end{equation}
where $\gamma$ is the optimal time warping between $f_1$ and $f_2$ (i.e. $f_1=(f_2\circ\gamma)\dot{\gamma}$).
\end{definition}\label{def:dist}

Notice that $d_{ext}(f_1, f_2) $ can also be written as $\Vert 1-\sqrt{F_2^{-1}\dot{\circ} F_1} \Vert=\Vert \sqrt{\dot{F}_1^{-1}}-\sqrt{\dot{F}_2^{-1}}\Vert$. 
To simplify the notation, we use $\gamma_f$ denote the cumulative distribution function of $f \in \mathbb P$.  Then 
$$ d_{ext}(f_1, f_2) = \Vert  \sqrt{\dot{\gamma}_{f_1}^{-1}}-\sqrt{\dot{\gamma}_{f_2}^{-1}} \Vert 
= \Vert  (1, {\gamma}_{f_1}^{-1}) - (1, {\gamma}_{f_2}^{-1}) \Vert 
$$
where the operator $(f,\gamma)=\left(f \circ \gamma\right)\sqrt{\dot{\gamma}}$ for $f \in \mathbb P$ and $\gamma \in \Gamma$.  

Also notice that because the optimal warping function $\gamma = F_2^{-1} \circ F_1$, the distance $d_\gamma(f_1, f_2) $ can also be written as $\Vert 1-\sqrt{F_2^{-1}\dot{\circ} F_1} \Vert=\Vert \sqrt{\dot{F}_1^{-1}}-\sqrt{\dot{F}_2^{-1}}\Vert$. The commonly-used Wasserstein distance $d_W$ is 
\begin{align*}
d_W(f_1,f_2) &=\left\Vert F_1^{-1} - {F}_2^{-1} \right\Vert 
\le 2 \left\Vert \sqrt{\dot{F}_1^{-1}}-\sqrt{\dot{F}_2^{-1}}\right\Vert 
= 2 d_\gamma(f_1,f_2)
\end{align*}
This shows that the consistency results that hold for $d_\gamma$ will also hold for $d_W$, but the reverse is not true in general. Similar to the Wasserstein and Hellinger distances, this $d_\gamma$ metric is also a proper distance.  The detailed proof in given in Appendix A. While $d_\gamma$ is not isometric like Bhattcharya and Hellinger, it is the only metric (within these four) that characterizes the phase difference between $f_1$ and $f_2$.

Either $d_{int}$ or $d_{ext}$ can be used to estimate the underlying density $f$.  In this paper, we choose to focus on the extrinsic distance $d_{ext}$ for two reasons: 1. Computational algorithms based upon the extrinsic distance are usually more efficient than those based on the intrinsic distance.  2. The extrinsic distance provides a closed-form Karcher mean representation (see definition next), which plays an essential role in developing the asymptotic theory for our estimator in Sec. \ref{sec:theory}.

\subsection{Karcher Mean}\label{sec:km}

The notion of a Karcher mean was used on the set of warping functions where an extrinsic distance between warping functions is adopted \cite{WuSrivastavaJCNS11}.  That is, assuming $\gamma_1, \cdots, \gamma_n \in \Gamma$ is a set of warping functions, their Karcher mean $\bar \gamma$ can be defined as 
$$\bar \gamma =   \argmin_{\gamma \in \Gamma} \sum_{i=1}^{n} ||\sqrt{\dot \gamma} - \sqrt{\dot \gamma_i}||^2.$$ 
It was shown in \cite{WuSrivastavaJCNS11} that this Karcher mean has a closed-form solution:
$$\sqrt{\dot {\bar \gamma}}=\dfrac{ \sum_{i=1}^n\sqrt{\dot{\gamma}_i}}{\Vert  \sum_{i=1}^n\sqrt{\dot{\gamma}_i}\Vert},$$ where $\sqrt{\dot {\bar \gamma}}$ is the SRVF of $\bar \gamma$.  \\

Similar to the Karcher mean of a set of warping functions in $\Gamma$, we can define the Karcher mean of a set of density functions in $\mathbb P$.  This definition is based on the newly-defined phase distance $d_{ext}$ in Eqn. \ref{eq:dis}.   
\begin{definition}
We define the Karcher mean $\mu_n$ of a set of functions $f_1, \cdots, f_n \in \mathbb P$ as the minimum of the sum of squares of distances in the following form:
\begin{equation}
\mu_n = \argmin_{\mu \in \mathbb P} \sum_{i=1}^{n} d_{ext}(\mu, f_{i})^{2}.
\label{eq:karcher}
\end{equation}
\end{definition}

Based on the closed-form solution for the Karcher mean of a set of warping functions, we can efficiently compute the Karcher mean in Eqn. \ref{eq:karcher} using the following algorithm. 

\subsubsection*{Algorithm 1: Karcher Mean Computation} \label{sec:kmalg}
Given a set of density functions $f_1, \dots,f_n \in \mathbb P$, and their cumulative distribution functions $F_1, \dots,F_n$, respectively. 
\begin{enumerate}
\item Initialize $f_0=f_j$ for any $j=1,2,\dots,n$.
\item Find $\gamma_j^*=F_0^{-1}\circ F_j,\,j=1,2,\dots,n.$ \label{item:gam}
\item Compute the Karcher mean $\bar \gamma$ of $\lbrace{\gamma_i^*}^{-1}\rbrace_{j=1}^n$, with formula 
$$\sqrt{\dot {\bar \gamma}}=\dfrac{ \sum_{j=1}^n\sqrt{\dot{\gamma_j^*}^{-1}}}{\Vert  \sum_{j=1}^n\sqrt{\dot{\gamma_j^*}^{-1}}\Vert}.$$
\item $\hat f =(f_0\circ\bar{\gamma}^{-1})\dot{\bar{\gamma}}^{-1}$ is the Karcher mean of $f_1, \dots,f_n$. 
\end{enumerate} 

The algorithm for computing the Karcher mean of functions in $\mathbb P$ is illustrated with a simple example in Fig. \ref{fig: betaex}. The 10 gold lines in the figure denote the density functions of Beta distribution on the domain  $[0,1]$ in the form $f(x; \alpha, \beta) \propto x^{\alpha-1} (1-x)^{\beta-1}$.  Here the parameters $\alpha$ takes value 1, 1, 1.5, 2, 2, 2.5, 3, 3, 4, 5, and $\beta$ takes value 4, 3, 3, 2.5, 2, 2, 1.5, 1, 1, 2 for the 10 functions, respectively.  The Karcher mean of these functions was computed using Algorithm 1 and the result is shown as the thick red line in Fig. \ref{fig: betaex}.

\begin{figure}[ht]
\begin{center}
\includegraphics[width=3.0in]{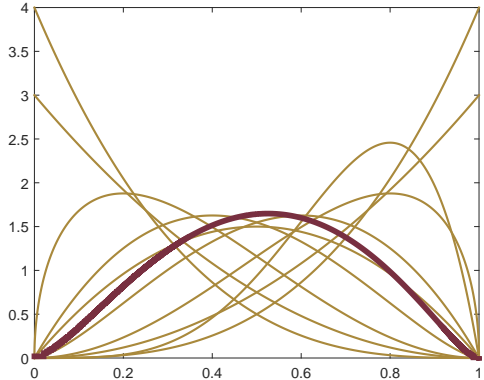}
\caption{Karcher mean of 10 Beta density functions}\label{fig: betaex}
\end{center}
\end{figure}

\subsection{Intensity Estimation Method}
\label{sec:estimation}
Since the Karcher mean of a set of density functions (computed under $d_\gamma$) is itself a density function, we use the Karcher mean as an estimate of the underlying density of the process. In our proposed estimation method, the Karcher mean, as computed using Algorithm 1, is used in conjunction with the MLE of the total intensity of the process to produce an estimate of the intensity function. Note that the computation for Karcher mean using Algorithm 1 is based on the assumption that each warped density $f_i(t)$ is already known, but practical data are only Poisson process realizations.  In this section, we propose a kernel estimation procedure to estimate $f_i(t)$. 

\subsubsection{Modified Kernel Density Estimation}\label{sec:mkde} 
Kernel density estimation has been well studied in statistics literature and it is well known that the standard kernel density estimator has good asymptotic properties when the domain is the real line. However, when the domain is a compact set such as [0, 1] in this paper, the standard kernel density estimator cannot be directly used. We adopt here a reflection-based method to address this issue \cite{clinehart91, schuster85, silverman86}.

Suppose $x_1,\dots , x_m$ are observations in $[0,1]$ whose density is given by $f \in \mathbb P$.
The standard kernel density estimator is given by $\tilde{f}(t)=\frac{1}{mh}\sumton{j=1}{m}K\left(\frac{t-x_j}{h}\right)$, where $K(\cdot)$ is a kernel function and $h$ denotes the kernel width.  Note that this estimated density is defined on
the real line $(-\infty, \infty)$, and the section within $[0, 1]$ in general is not a density function itself.  To simplify the estimation procedure, we can choose kernel functions with compact support within $[-1, 1]$. That is, $K(t) = 0$ for $|t| > 1$.


Here we propose a two-step modification of the estimate $\tilde f(t)$.  At first, we wrap around $\tilde f$ within the domain $[0, 1]$, and denote the new function as $\tilde {\tilde f}$, a density function on $[0, 1]$.  Secondly, we add a small positive constant to $\tilde {\tilde f}$, and then normalize the sum to be a density function.  This step is to assure that the normalized function is positive on $[0, 1]$, a necessary condition for the existence of the warping functions used in the distance $d_{ext}$.  The modified kernel density estimation can be summarized in the following algorithm. 

\subsubsection*{Algorithm 2: Modified Kernel Estimation} \label{sec:fialg}
Suppose $x_1,\dots , x_m$ are observations in $[0,1]$ whose density is given by $f$. 

\begin{enumerate}
\item Calculate the standard kernel-based estimate, $\tilde{f}(t)=\dfrac{1}{mh}\sumton{j=1}{m}K\left(\dfrac{t-x_j}{h}\right)$, $t\in \real$, using an appropriate bandwidth, $h$, and a kernel function $K$ with compact support (e.g. a Beta density function). 
\item Update the estimate by $\tilde{\tilde{f}}(t)=\tilde{f}(t)+\tilde{f}(-t)+\tilde{f}(2-t)$. The updated estimate is defined only for $t\in [0,1]$. 
\item $\hat{f}(t)= \tilde{\tilde{f}}(t)\dfrac{m}{m+1}+\dfrac{1}{m+1}, t \in [0, 1]$ is the modified estimate.
\end{enumerate}

\subsubsection{Estimation Algorithm} 
Estimation of the intensity of the process occurs in two independent components. First, the total intensity $\Lambda$ can be easily computed with a standard MLE procedure. Second, the Karcher mean of the estimated densities is used to estimate $f(t)$. This estimation algorithm is given as follows.

\subsubsection*{Algorithm 3: Intensity Estimation Algorithm } 
Given a set of observed processes $S_i$ with number of events being $k_i, i=1,\dots,n$, 

\begin{enumerate}
\item Estimate $\Lambda$ by its MLE: $\hat{\Lambda}=\frac{1}{n}\sumton{i=1}{n}{k_i}$.
\item Use Algorithm 2 to estimate the density of each observed process, $\hat f_i(t), i = 1, \cdots, n$.   
\item Estimate the intensity function of each process by the formula $\hat \lambda_i(t) = \hat \Lambda \hat f_i(t), i = 1, \cdots, n$.  
\item Use Algorithm 1 to estimate the overall underlying density, $\hat f(t)$, as the Karcher mean of $\{\hat{f_i}\}$.
\item Estimate the underlying intensity $\lambda(t)$ in the original process using:
\begin{center}
$\hat{\lambda}(t)=\hat{\Lambda} \hat{f}(t)$.
\end{center}
\end{enumerate}

\section{Asymptotic Theory on Consistency}
\label{sec:theory}
Asymptotical properties of estimators are often of interest since these properties can give reasonable certainty that the ground-truth parameters are appropriately estimated by the given algorithms. In this section, we provide asymptotic theory on the density estimator $\hat f$ in Algorithm 3.  Our estimation is based on the model 
$$\lambda_i = (\lambda \circ \gamma_i) \dot \gamma_i,  \ \  i = 1, \cdots, n,$$ 
where $\lambda$ is the underlying intensity function and $\gamma_i \in\Gamma,\, i=1,\dots,n,$ are a set of warping functions.  By Lemma \ref{lem:int_warp}, each observation $S_i$ is a Poisson process realization with intensity $\lambda_i$.  Given $\{S_i\}$, Algorithm 3 provides an estimation procedure for $\lambda$.  As the total intensity $\Lambda$ is independent of time warpings,  our asymptotical theory will focus on the normalized intensity, i.e. intensity function $f = \lambda/\Lambda$.  We mathematically prove that the proposed algorithm provides a consistent estimator for $f$.  The asymptotic theory is based on sample size $n$ as well as the total intensity $\Lambda$.  Here we only provide result on the main theorem.  All lemmas that lead to the theorem can be found in Appendix B. 

Before we state the main theorem, we list all assumptions as follows:
\begin{enumerate}
\item The observations are a sequence of Possion process realizations $\{S_i\}$, and $S_i$ follows intensity function $\lambda_i = (\lambda \circ \gamma_i) \dot \gamma_i$.   $\Lambda = \int_0^1 \lambda(t)dt$ is the total intensity.  $f = \lambda/\Lambda$ and $f_i = \lambda_i/\Lambda = (f; \gamma_i)$.
\item The density function $f$ is continuous on [0,1].  Also, there exist $m_f, M_f > 0$ such that $f(t) \in [m_f, M_f]$, for any $t \in [0,1]$.   
\item $\gamma_i(t), t \in [0, 1], i = 1, \cdots, n$ are a set of independent warping functions.  The SRVFs of their inverses $\sqrt {\dot \gamma_i^{-1}(t)}$ distribute around $\sqrt {\dot \gamma_{id}} = 1$ on the Hilbert unit sphere $H^\infty$. In particular, $E(\sqrt {\dot \gamma_i(t)}) \equiv \beta > 0$ and there exist $m_\gamma, M_\gamma > 0$ such that
$\sqrt {\dot \gamma_i^{-1}(t)} \in [m_\gamma, M_\gamma]$, for any  $t \in [0, 1]$. It is important to note that $\sqrt{\dot{\gamma}_i^{-1}(t)}$ is a point on the Hilbert unit sphere. As a result, it is easy to show that assuming $\mathbb{E}\left(\sqrt{\dot{\gamma}_i^{-1}(t)}\right)=\beta>0$ is equivalent to assuming that the extrinsic mean of $\{\sqrt{\dot{\gamma}_i^{-1}(t)}\}$ is 1.
\item The total intensity $\Lambda$ can vary in the form of a sequence $\{\Lambda_m\}_{m = 1}^\infty$. We assume the sequence goes to $\infty$ with $\Lambda_m \ge \alpha \log(m), \alpha > 1$ for  sufficiently large $m$.
\item The bandwidth of the kernel density estimator in Algorithm 2 is chosen optimally. That is, for a sequence of $r$ events, the bandwidth $h_r$ satisfies $h_r \rightarrow 0$ and $r h_r \rightarrow \infty$ when $r \rightarrow \infty$.
\end{enumerate}

\begin{theorem}
Given the four conditions listed above, let $\hat f$ be the density function estimated with Algorithm 3.  Then we have 
$$ \lim_{n \rightarrow \infty} \lim_{m \rightarrow \infty} d_{ext}(\hat f, f) = 0 \ \ \ a.s. $$ 
\end{theorem}

\begin{proof} By the basic property of a Poisson process, the event times in the observation $S_i$ are an i.i.d. sequence with density function $f_i = (f; \gamma_i), i = 1, \cdots, n$.   
Denote $\hat f_i$ as the estimated density function by the modified kernel estimation method.  Then $\hat f_i(t) > 0$ for any $t \in [0, 1]$.  Based on the group structure of $\Gamma$, there exists a unique $\hat \gamma_i \in \Gamma$ such that $\hat f_i = (f; \hat \gamma_i)$.  

Here we compute the Karcher mean of $\{\hat f_i\}$.  For any density function $g$, we have  
\begin{eqnarray*}
\sum_{i=1}^{n}{d_{ext}^2(\hat f_i,g)} & = & \sum_{i=1}^{n}{d_{ext}^2((f;\hat \gamma_i),g)} 
 = \sumton{i=1}{n}{\left\Vert \left(1, \gamma_{(f; \hat \gamma_i)}^{-1} \right) - \left(1, \gamma_{g}^{-1}\right) \right\Vert^2} \\
& = & \sumton{i=1}{n}{\left\Vert \left(1, \hat \gamma_i^{-1} \circ \gamma_{f}^{-1} \right) - \left(1, \gamma_{g}^{-1}\right) \right\Vert^2} \\
&=&  \sumton{i=1}{n}{\left\Vert \left(1, \hat \gamma_i^{-1}  \right) - \left(1, \gamma_{g}^{-1} \circ \gamma_{f} \right) \right\Vert^2} 
 \end{eqnarray*}
Denote the Karcher mean of $\{\hat \gamma_i^{-1}\}$ as $\hat {\bar{\gamma}}$.  Then the above sum of squares is minimized when 
$\gamma_{g}^{-1} \circ \gamma_{f} = \hat {\bar \gamma}$.  That is, 
$\gamma_{\hat f}^{-1} = \hat {\bar \gamma} \circ \gamma_f^{-1}.$
By isometry on time warping functions and the triangular inequality,
\begin{eqnarray*}
d_{ext}(\hat f, f) & = & {\left\Vert \left(1, \gamma_{\hat f}^{-1} \right) - \left(1, \gamma_{f}^{-1}\right) \right\Vert} 
                    = {\left\Vert \left(1, \hat {\bar \gamma} \circ \gamma_f^{-1} \right) - \left(1, \gamma_{f}^{-1}\right) \right\Vert} \\
                & = &  {\left\Vert \left(1, \hat {\bar \gamma} \right) -1 \right\Vert} 
 		 \le  {\left\Vert \left(1, \hat {\bar \gamma} \right) - \left(1, \bar \gamma \right)\right \Vert}  + {\left\Vert \left(1, \bar \gamma \right) -1 \right \Vert}                
 \end{eqnarray*}
By Lemma \ref{lem:karcher}, we have shown that  ${\left\Vert \left(1, \bar \gamma \right) -1 \right\Vert}  \xrightarrow{a.s.} 0$ when $n \rightarrow \infty$.  Note that  ${\Vert \left(1, \hat {\bar \gamma} \right) - \left(1, \bar \gamma \right) \Vert}$ depends on the total intensity $\Lambda_m$ and sample size $n$.  We will show that this term also converges to 0 when $m$ is large (for any fixed $n$).  

To simplify the notation, we denote $a_i = \sqrt{\dot \gamma_i^{-1}}$, $\hat a_i = \sqrt{\dot {\hat\gamma}_i^{-1}}, i = 1, \cdots, n.$   Let the number of events in $S_i$ be $n_i$.  Then $n_i$ is a random variable following Poisson distribution with mean $\Lambda_m$. By Lemma  \ref{lem:poisson}, $n_i \xrightarrow{a.s.} \infty$ when $m \rightarrow \infty$.  Using Lemma \ref{lem:inverse}, 
\begin{eqnarray*}
\left\Vert\hat a_i - a_i\right \Vert &=& \left\Vert (1, \hat \gamma_i^{-1}) - (1, \gamma_i^{-1})\right \Vert =  
 \left\Vert (1, \hat \gamma_i^{-1} \circ \gamma_f^{-1}) - (1, \gamma_i^{-1} \circ \gamma_f^{-1}) \right \Vert \\
&=&  \left\Vert (1, \gamma_{(f; \hat \gamma_i)}^{-1}  ) - (1, \gamma_{(f; \gamma_i)}^{-1} ) \right \Vert
=  \left\Vert (1, \gamma_{\hat f_i}^{-1}  ) - (1, \gamma_{f_i}^{-1} ) \right \Vert
 \xrightarrow{a.s.} 0
 \end{eqnarray*}
  when $n_i \rightarrow \infty$.  Therefore,  $||\hat a_i - a_i|| \xrightarrow{a.s.} 0, \ i = 1, \cdots, n,$ when $m \rightarrow \infty$.

Let $\bar {\hat a} = \frac{1}{n}\sum_{i=1}^n \hat a_i$ and $\bar a = \frac{1}{n}\sum_{i=1}^n a_i$.  Then, $||\bar {\hat a} - \bar a|| \xrightarrow{a.s} 0$ when $m \rightarrow \infty$.  Hence, 
\begin{eqnarray*}
{\left\Vert \left(1, \hat {\bar \gamma} \right) - \left(1, \bar \gamma \right) \right \Vert}
&=& \left\Vert \frac{\bar {\hat a}}{\left\Vert\bar {\hat a}\right \Vert} -  \frac{\bar { a}}{\left\Vert\bar {a}\right \Vert} \right \Vert 
\le \left\Vert \frac{\bar {\hat a}}{\left\Vert\bar {\hat a}\right \Vert} -  \frac{\bar {\hat a}}{\left\Vert\bar { a}\right \Vert} \right \Vert + \left\Vert \frac{\bar {\hat a}}{\left\Vert\bar { a}\right \Vert} -  \frac{\bar { a}}{\left\Vert\bar {a}\right \Vert} \right \Vert \\
&\le& 2 \left\Vert \bar {\hat a} - \bar {a} \right \Vert / \left\Vert\bar {a}\right \Vert  \xrightarrow{a.s.} 0 \ \ (\mbox{when } m \rightarrow \infty)
\end{eqnarray*}
Note that the convergence of ${\Vert \left(1, \hat {\bar \gamma} \right) - \left(1, \bar \gamma \right) \Vert}$ is for any sample size $n$.  Finally, we have proved that 
 $$ \lim_{n \rightarrow \infty} \lim_{m \rightarrow \infty} d_{ext}(\hat f, f) = 0 \ \ \ a.s. $$ 
\end{proof}

\section{Extension to Nonnegative Intensity Functions}
\label{sec:nonnegdf}
The method developed thus far applies only to strictly positive density functions. In practice, this may be a quite restrictive condition and it is desired to extend the method to non-negative density functions.  Our estimation is still based on the model 
$$\lambda_i = (\lambda \circ \gamma_i) \dot \gamma_i,  \ \  i = 1, \cdots, n,$$ 
where $\lambda \ge 0$ is the underlying intensity function and $\gamma_i \in\Gamma,\, i=1,\dots,n,$ are a set of warping functions. In this section, we propose to extend Algorithm 3 to estimate this nonnegative $\lambda$ with Poisson process observations. 

\subsection{Representation of Nonnegative Intensities} 
For estimation, our focus is still on the density function $f = \lambda/\Lambda$ as the total intensity $\Lambda$ is independent of the time warping.  Let $F$ denote the CDF of $f$.  Then $F(0) = 0, F(1)= 1.$ However, as $f$ is nonnegative, $F$ may not be strictly increasing on the domain $[0,1]$.   To simplify the representation, we assume that $F$ is strictly increasing except being constant on a finite number, $K$, of non-overlapping intervals (This finiteness assumption would be 
sufficient for nonnegative intensities in practical use).  Let $\mathcal F = \{F \circ \gamma | \gamma \in \Gamma\}$ denote the set of CDFs which are warped versions of $F$, and $F_i$ be the CDF of $f_i = \lambda_i/\Lambda$.  Then $F_i = F \circ \gamma_i \in \mathcal F$ will also be constant on corresponding intervals.  

In general, let $h, g$ be two density functions whose CDFs $H, G$ are in $\mathcal F$.  Then $H$ and $G$ are strictly increasing except being constant on $K$ non-overlapping intervals.  We define 
 $\Gamma_{h, g}=\{\gamma\in \Gamma | h = ( g\circ\gamma)\dot{\gamma}\}=\{\gamma\in \Gamma|H=G\circ\gamma\}$. 
 By construction, $\Gamma_{h,g} \neq \emptyset$.  We denote the $K$ constant intervals for $H$ and $G$ are 
 $[a_1,\,b_1],\cdots,[a_K,\,b_K]$ and $[c_1,\,d_1],\cdots,[c_K,\,d_K]$, respectively.  
For any $\gamma\in \Gamma_{h,\,g}$, we must have $\gamma(a_k)=c_k$ and $\gamma(b_k)=d_k$ for $k=1,\cdots,K$.  To include the boundary points, we denote $b_0=d_0=0$ and $a_{K+1}=c_{K+1}=1$. It is our goal  to characterize all warping functions in $\Gamma_{h,\,g}$.

Note that the function $G$ is strictly increasing on each interval $[d_k,\,c_{k+1}], k = 0, 1, \cdots, K$.  Now
we define a mapping $G_k:[d_k,\,c_{k+1}]\rightarrow\real$ as follows,
$$ G_k(s) = G(s), s \in [d_k, c_{k+1}]. $$
It is apparent that $G_k$ is strictly increasing on its domain $[d_k, c_{k+1}], k = 0, 1, \cdots, K$.  For any $\gamma \in \Gamma_{h,g}$ and $t \in [b_k, a_{k+1}]$,  $\gamma(t)$ is in $[d_k, c_{k+1}]$.  Hence, $H(t) = G(\gamma(t)) = G_k(\gamma(t))$, and $\gamma(t) = G_k^{-1} \circ H(t)$. 

We then focus on the regions $[c_k,\,d_k],\,k=1,\cdots,K$ where  $G$ is constant (note: $G^{-1}$ does not exist). 
Note that $H(a_k) = G(\gamma(a_k)) = G(c_k) = G(d_k) = G(\gamma(b_k)) = H(b_k)$.  Hence, any $\gamma \in \Gamma$ with $\gamma(a_k) = c_k, \gamma(b_k) = d_k$ satisfies that $H(t) = G(\gamma(t))$ for any $t \in [a_k, b_k]$.   Finally, we have shown that the set $\Gamma_{h,g}$ can be characterized as follows,
\begin{eqnarray*}
& & \Gamma_{h,g} = \{\gamma \in \Gamma | \gamma(t) = G_k^{-1} \circ H(t), t \in [b_k, a_{k+1}], k = 0, \cdots, K, \\
 & & \hspace{3cm} \gamma(a_k) = c_k, \gamma(b_k) = d_k, k = 1, \cdots, K\}.
\end{eqnarray*}

\subsection{Estimation of Nonnegative Intensities} 

In Sec. \ref{sec:method}, we defined a phase distance $d_{ext}$ between two positive density functions.  Here we generalize the distance to nonnegative densities. 

\begin{definition}
Let $h, g$ be two density functions whose CDFs $H, G$ are in $\mathcal F$.
We define the distance between $h$ and $g$ as 
\begin{equation}
D(h,\,g)= \underset{\gamma\in\Gamma_{h,g}}{\inf}\Vert 1-\sqrt{\dot{\gamma}}\Vert
\label{eq:D}
\end{equation}
\end{definition}

We present three properties of this distance below. 
\begin {enumerate}
\item $D$ is a generalization of the distance $d_{ext}$ -- for strictly positive densities $h, g$, the set $\Gamma_{h, g}$ has single element $G^{-1}\circ H$, and therefore  $D(g, h) = d_{ext}(g, h)$. 
\item $D$ is a proper distance. The proof of this property is similar to that for the distance $d_{ext}$ (see Appendix A) and is, therefore, omitted here.  
\item Denote the constant intervals for $H$ and $G$ as 
 $[a_1,\,b_1],\cdots,[a_K,\,b_K]$ and $[c_1,\,d_1],\cdots,[c_K,\,d_K]$, respectively. Then the infimum of $\Vert 1-\sqrt{\dot{\gamma}}\Vert$ over $\Gamma_{h,g}$ can be  \textit{uniquely} reached.  Specifically, let 
 \[\gamma^*(t) =  \underset{\gamma\in\Gamma_{h,g}}{\arginf} \Vert 1-\sqrt{\dot{\gamma}}\Vert \] 
Then, 
\begin{equation}
\gamma^*(t) = \left\{
  \begin{array}{lr}
    G_k^{-1}\circ H(t) &  t\in [b_k,\,a_{k+1}],\,k=0,1,\cdots,K\\
    \left(\dfrac{d_k-c_k}{b_k-a_k}\right)(t-a_k)+c_k &  t\in [a_k,\,b_k], \, k=1,\cdots,K
  \end{array}  
\right.
\label{eq:opt}
\end{equation}
The proof of this property is based on the following fact (shown in \cite{spikemeanwu}): Assume $\gamma$ is a mapping in $\Gamma_0 = \{\gamma: [a,b] \rightarrow [c,d] | \gamma(a) = c, \gamma(b) = d, \dot \gamma(t) > 0, t \in [a,b]\}$. Then,  
the distance $\Vert 1-\sqrt{\dot{\gamma}}\Vert$ is minimized over $\Gamma_0$ when $\gamma$ is a linear function from $[a,b]$ to $[c,d]$. 
\end {enumerate}

{\bf Estimation Method:}  The estimation of nonnegative intensities follows the same procedure as in the Intensity Estimation Algorithm (Algorithm 3), where Algorithm 1 calls for the Karcher mean computation.  However, in this case we need to update the second step of Algorithm 1 (computation of optimal warping between $F_0$ and $F_j$), the new optimal form in Eqn. \ref{eq:opt} is adopted.  Analogous to the proof in Sec. \ref{sec:theory}, one can demonstrate that the estimated nonnegative intensity is also an consistent estimator (under the metric $D$ in Eqn. \ref{eq:D}).  We omit the details in this manuscript to avoid repetition. 
 
\section{Experimental Results}\label{sec:results}
In this section we will demonstrate the proposed intensity estimation using two simulations -- one is for a strictly positive intensity, and the other is for an intensity with zero-valued sub-regions.  We will also apply the new method in a real spike train dataset and evaluate the classification performance using the estimated intensities.  

\subsection{Simulations for Illustration}
\subsubsection{Poisson Process with a Positive Intensity Function}
Twenty independent realizations of a non-homogeneous Poisson process were simulated with the intensity function $\lambda(t)=100(3 + 2\sin((8t-1/2)\pi))$ on [0, 1].  This intensity function and these 20 original processes are shown in Fig. \ref{fig:intest}A. Because of the non-constant intensity, there is a higher concentration of events during intervals with high intensity and fewer events during intervals with low intensity. This pattern is easily seen in the simulated processes. 

\begin{figure}[ht]
\begin{center}
\begin{tabular}{ccc}
 \hspace{-24pt} \textbf{A} &  \hspace{-24pt} \textbf{B} & \hspace{-24pt}  \textbf{C}\\
\hspace{-24pt}
\includegraphics[height=1.4in]{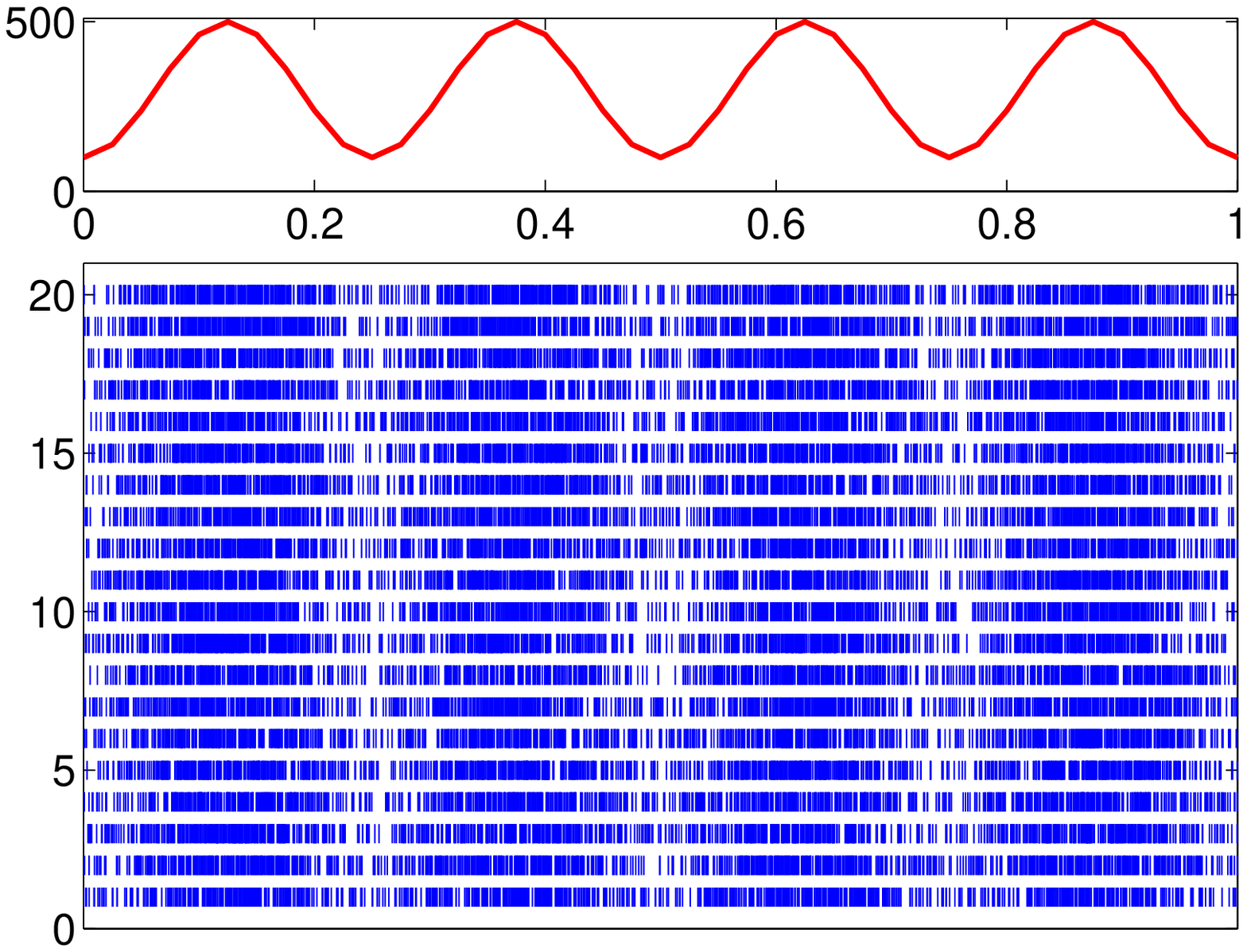}&   \hspace{-24pt}
\includegraphics[height=1.4in]{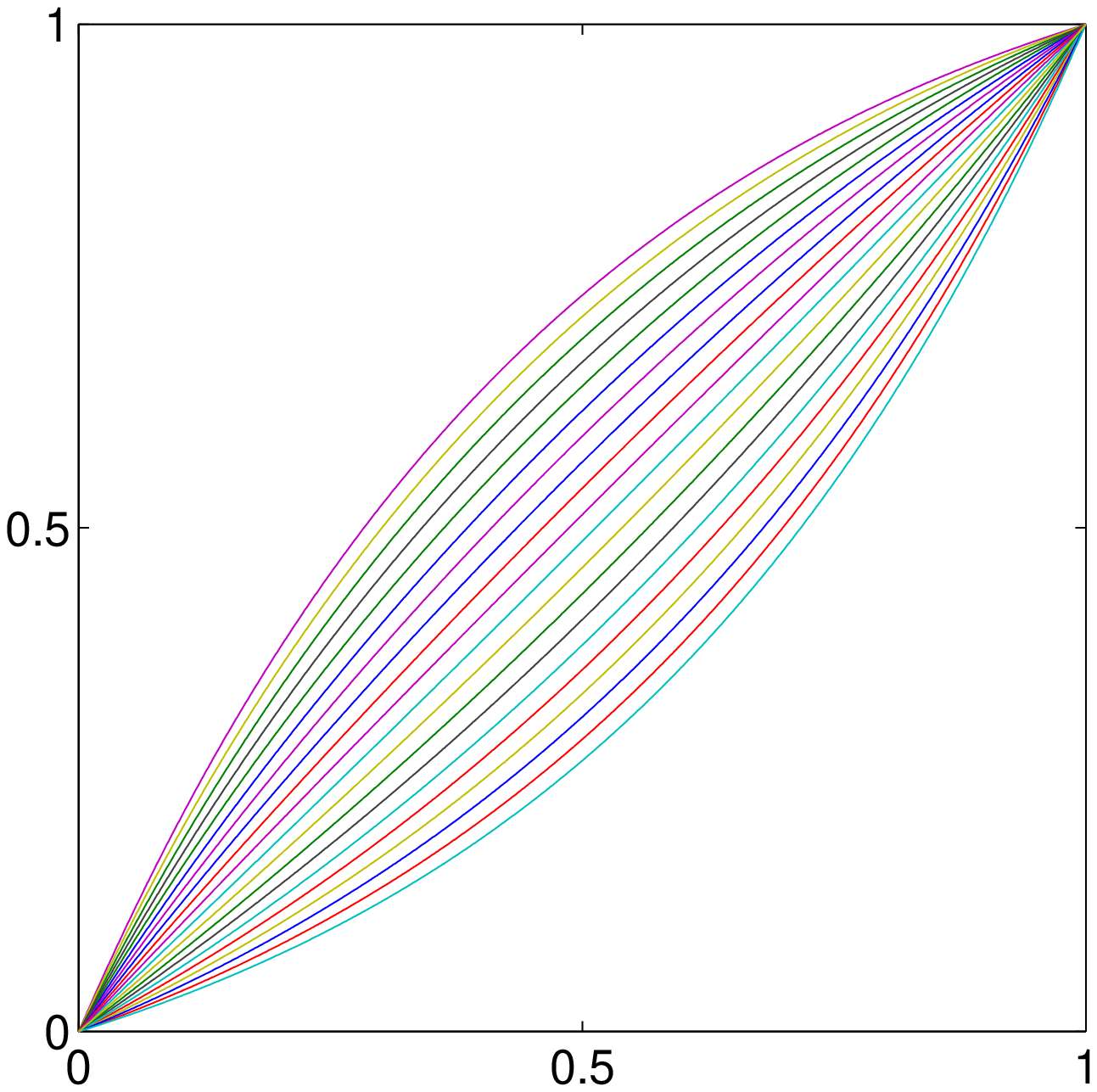}&  \hspace{-24pt}
\includegraphics[height=1.4in]{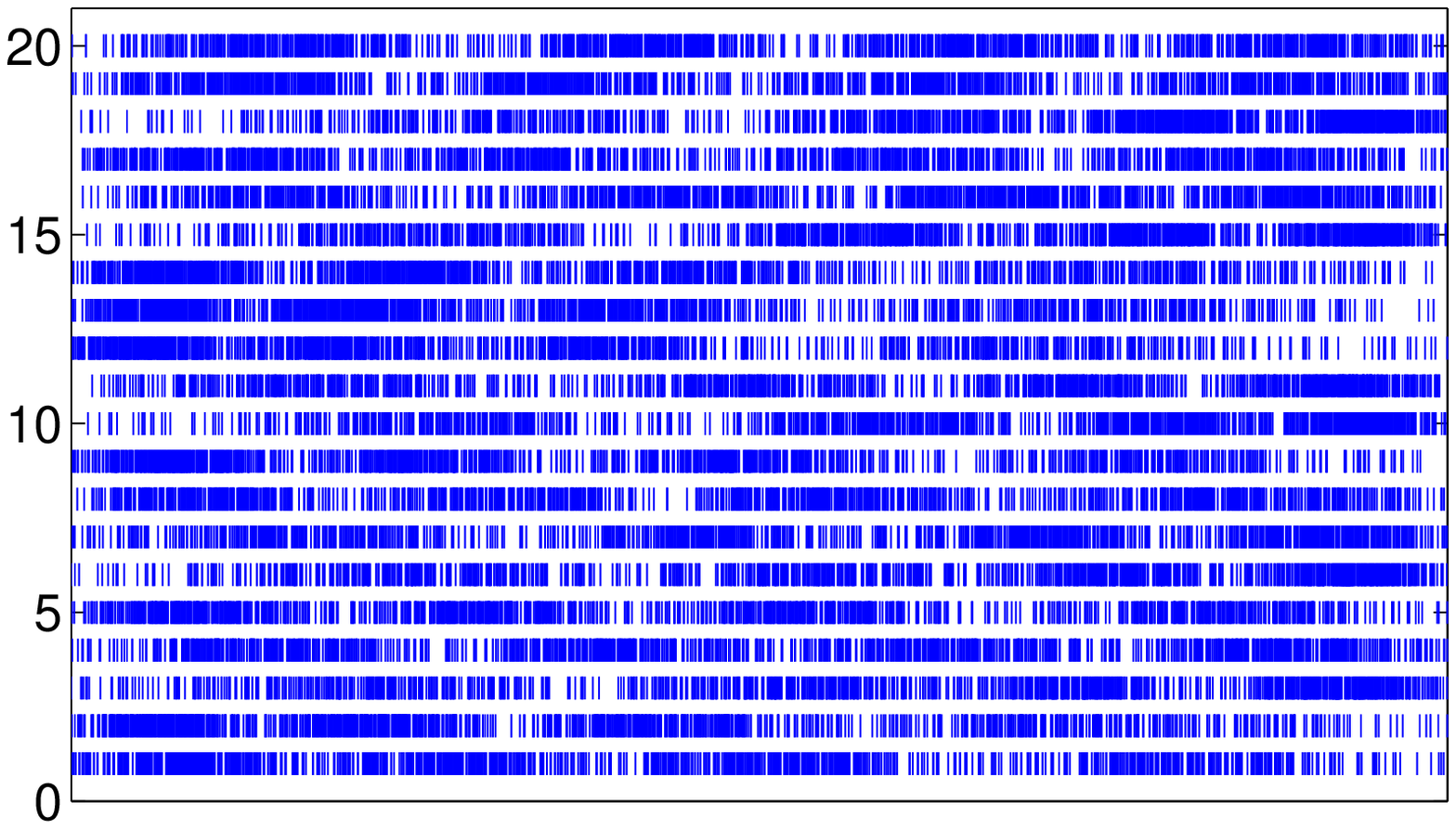}
\end{tabular}
\caption{Simulation of Poisson process with compositional noise. \textbf{A}. Intensity function of a Poisson process (top panel) and 20 independent realizations (bottom panel).  \textbf{B}. 20 time warping functions. \textbf{C}. 20 observed processes, which are warped version of the original 20 Poisson process realizations.} \label{fig:intest}
\end{center}
\end{figure}

We then generate 20 warping functions $\{\gamma_i\}_{i = 1}^{20}$ in the following form: 
 $\gamma_i(t) = {e^{a_i t} -1 \over e^{a_i} - 1}$. 
Here $a_i$ are equally spaced between $-2$ and $2$, $i = 1, \cdots, 20$.   These warping functions are shown in Fig. 
 \ref{fig:intest}B.  We then warp the 20 independent Poisson process using these 20 warping functions, respectively, by 
the formula in Eqn. \ref{eq:warping}.  The resulting warped processes are shown in Fig. \ref{fig:intest}C. 
Comparing these processes with those in  Fig. \ref{fig:intest}A, we can see that the clear link between number of events in each sub-region and the intensity value no longer exists.  Given these noisy Poisson process observations, we aim to reconstruct the underlying intensity function $\lambda(t)$.  


The individual estimated density functions for the warped processes are shown in the top panel of Figure \ref{fig:simintenest}A. The true warped density functions are shown in the bottom panel of Figure \ref{fig:simintenest}A. The underlying intensity function was estimated for two different cases. In the first case, time warping is present and ignored during estimation. In the second case, time warping is present and accounted for in the estimation using the proposed method. Both of these estimates are displayed with the true intensity function for comparison in Figure \ref{fig:simintenest}B. When time warping is present and ignored, the estimated intensity function underestimates the true intensity in the middle two-thirds of the curve and the true pattern is not revealed. However, when the warping is accounted for during the estimation process, the estimate is a much better estimate of the true intensity function. When the warping functions are more severe (shown in Figure \ref{fig:simintenest}C; $a_i\in [-4,4]$), the performance decreases in all methods (Figure \ref{fig:simintenest}D). The $\lone$-, $\ltwo$-, and $\mathbb{L}^\infty$- norms were all used to measure the error in estimating the true intensity for each method (Table \ref{table:kmerrors}). However, the proposed method consistently has the lowest error regardless of which norm is used to measure the error. 

\begin{figure}[ht!]
\begin{center}
\begin{tabular}{cc}
\textbf{A.} & \textbf{B.} \\
 \includegraphics[height=2.0in]{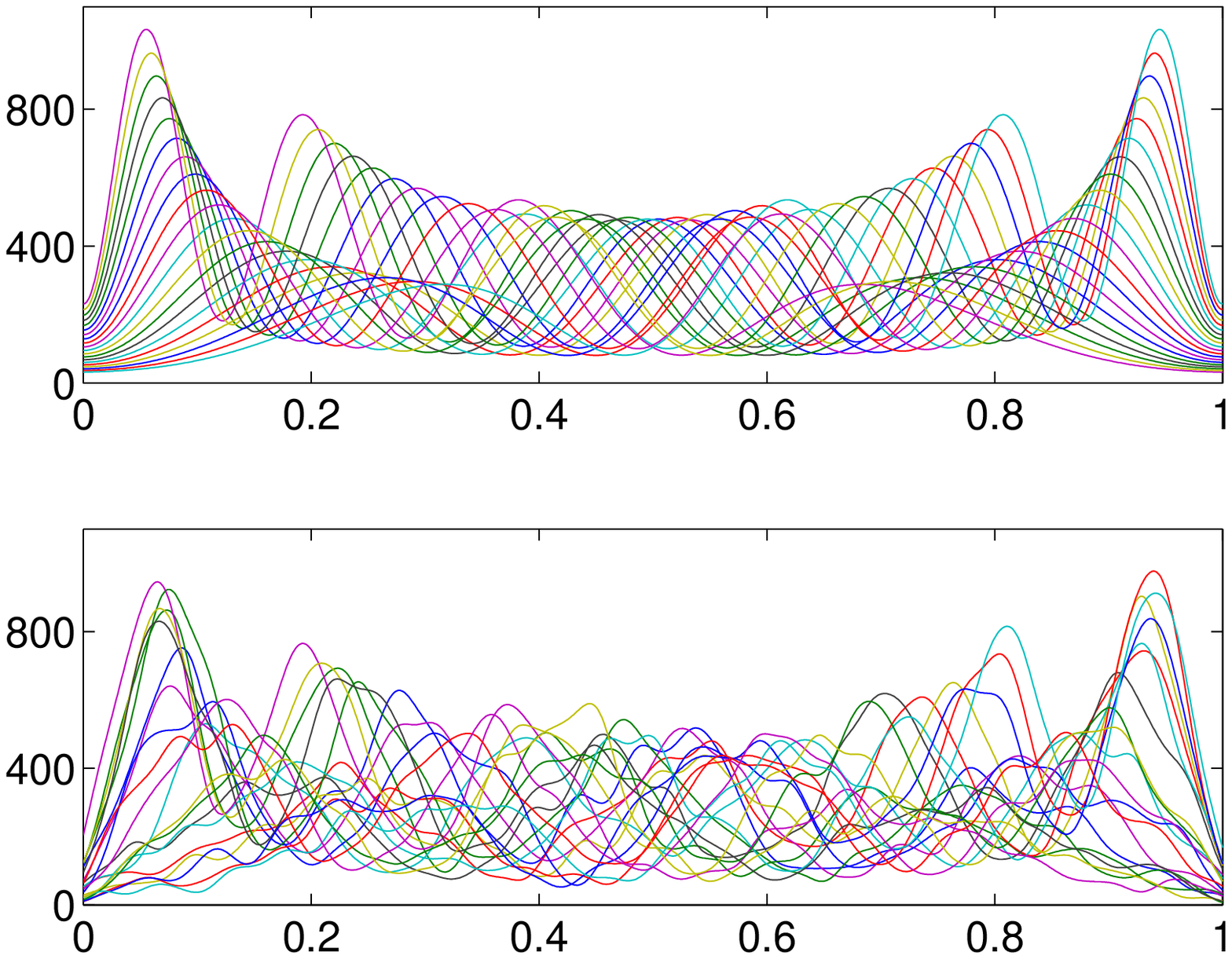}&
 \includegraphics[height=2.0in]{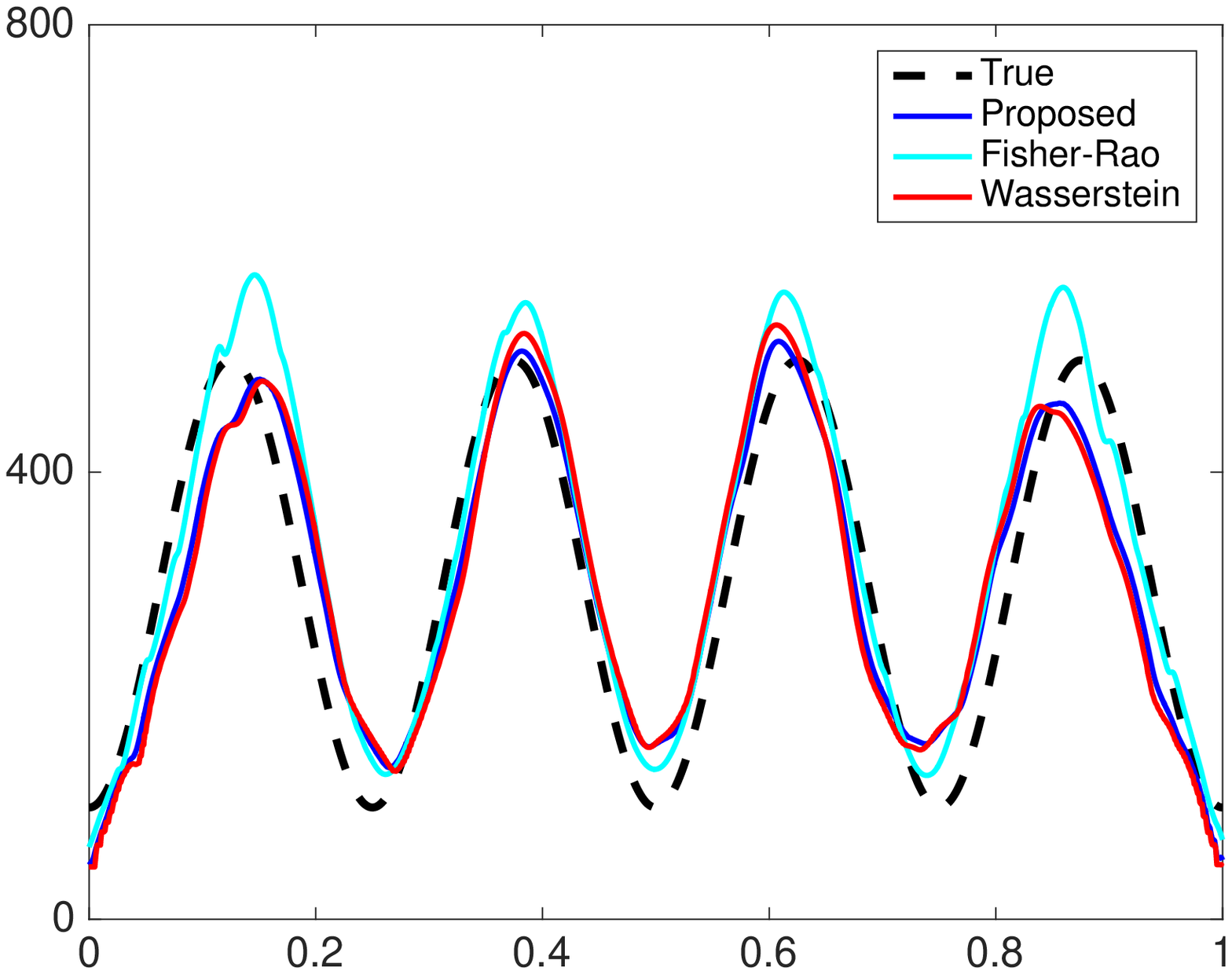}\\
 \textbf{C.} & \textbf{D.} \\
\includegraphics[height=2.0in]{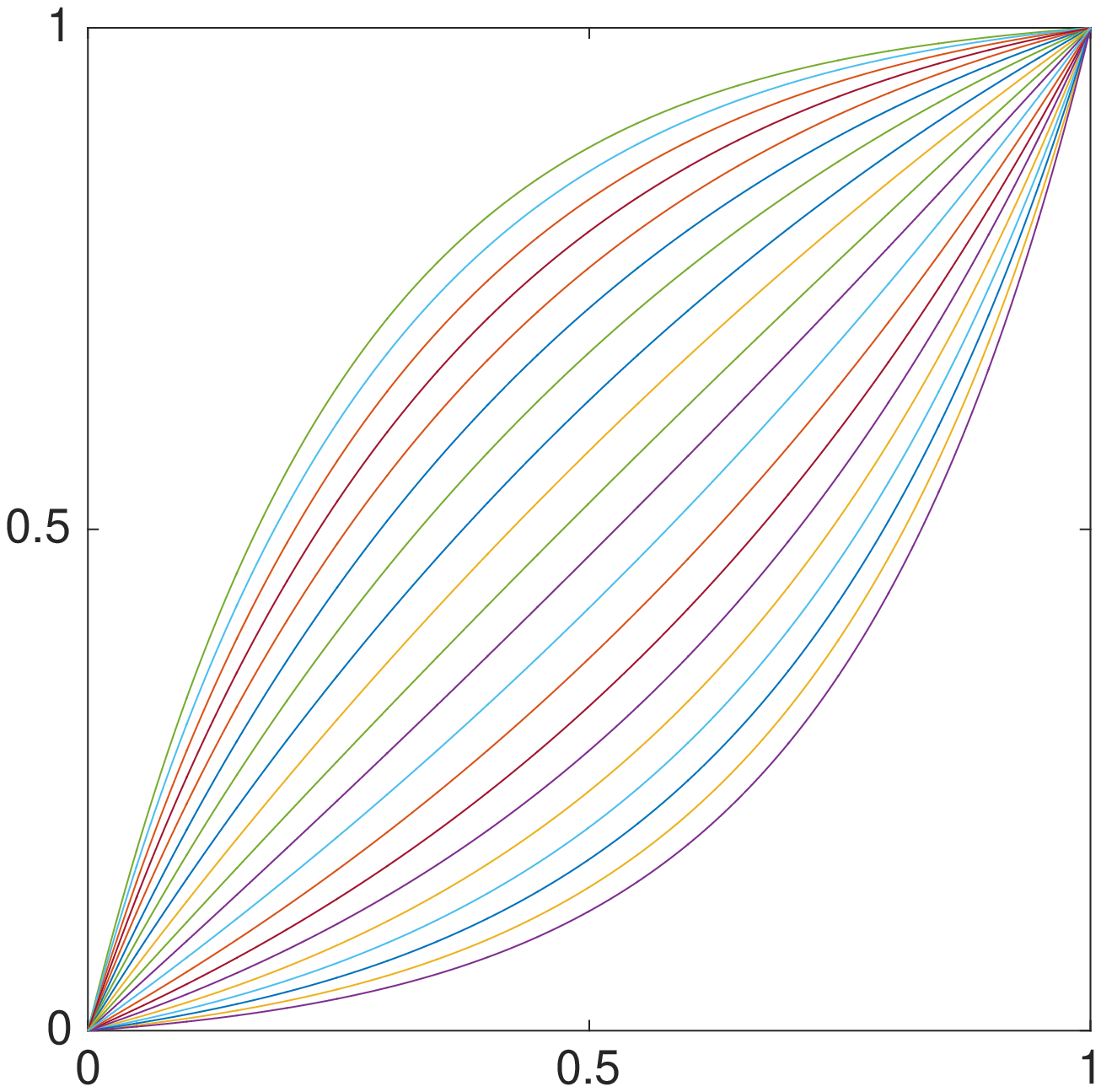}&
 \includegraphics[height=2.0in]{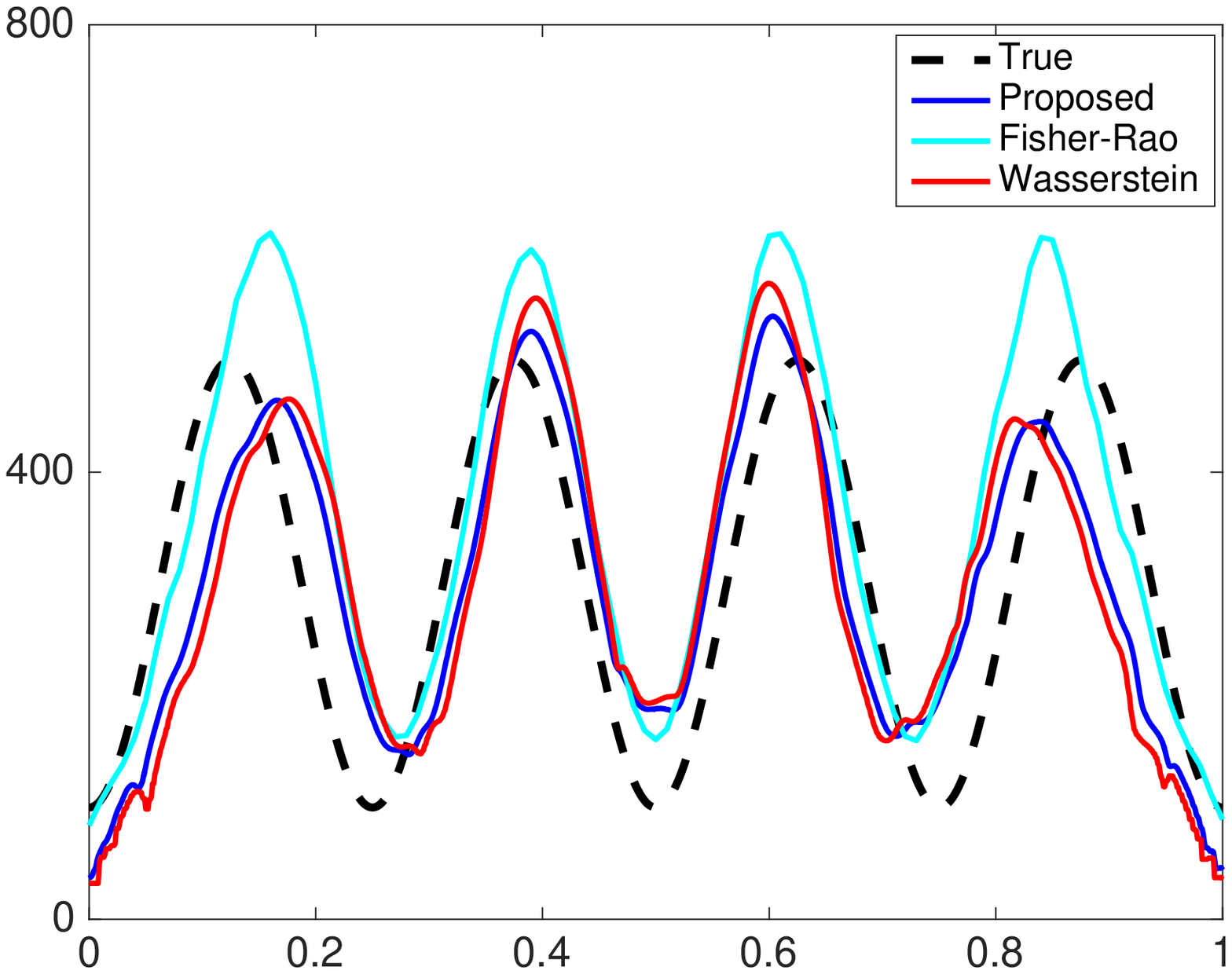}\\

\end{tabular}
\caption{Warped Poisson Process. \textbf{A.} Top panel: estimated individual density functions for warped processes. Bottom panel: True warped density functions. \textbf{B.} Estimated intensity functions computed under four methods. \textbf{C.} Warping functions for second simulation. \textbf{D.} Updated intensity estimates computed under four methods with more severe warping functions from panel \textbf{C}.} \label{fig:simintenest}

\end{center}
\end{figure}

%

\begin{table}[ht]
\begin{center}
\caption{Three types of errors for each method under 2 sets of time warping functions}\label{table:kmerrors}
\begin{tabular}{cc|ccc} 
Time  & Norm & Proposed & Fisher-Rao & Wasserstein \\ 
 Warping & & & & \\ \hline
\multirow{3}{*}{1} & \lone & \bf{81177.1} & 84372.9 & 105504.9 \\
 & \ltwo & \bf{2876.0} & 3369.6 & 3715.9 \\
 & $\mathbb{L}^\infty$ & \bf{166.1} & 243.1 & 207.5 \\ \hline
 \multirow{3}{*}{2} & \lone & \bf{117991.4} & 129957.4 & 155215.8\\
  & \ltwo & \bf{4184.3} & 5377.5 & 5504.3 \\ 
   & $\mathbb{L}^\infty$ & \bf{228.0} & 381.7 & 301.5 \\ \hline
\end{tabular}
\end{center}
\end{table}

\subsubsection{Poisson Process with a Nonnegative Intensity Function}

In this second example, we illustrate the estimation method for non-negative intensity functions in Sec. \ref{sec:nonnegdf}.  
The underlying intensity function is defined on $[0,\,1]$ and given in the following form:
\[\lambda(t) = \left\{
  \begin{array}{lr}
    -16000|t-0.5|+4000 &  t\in [0.25, \, 0.75]\\
    \qquad \quad 0 &  \mathrm{otherwise}
  \end{array}
\right.
\]
This intensity, shown in Fig. \ref{fig:trisim}A, has a trianglar shape with two flat sub-regions, $[0,\,0.25]$ and $[0.75,\,1]$, which occur on either side of the triangle whose peak is located at $t=0.5$. 

\begin{figure}[ht!]
\begin{center}
\begin{tabular}{ccc}
 \hspace{-24pt} \textbf{A} &  \hspace{-24pt} \textbf{B} &  \hspace{-24pt} \textbf{C}\\  \hspace{-24pt}
 \includegraphics[height=1.3in]{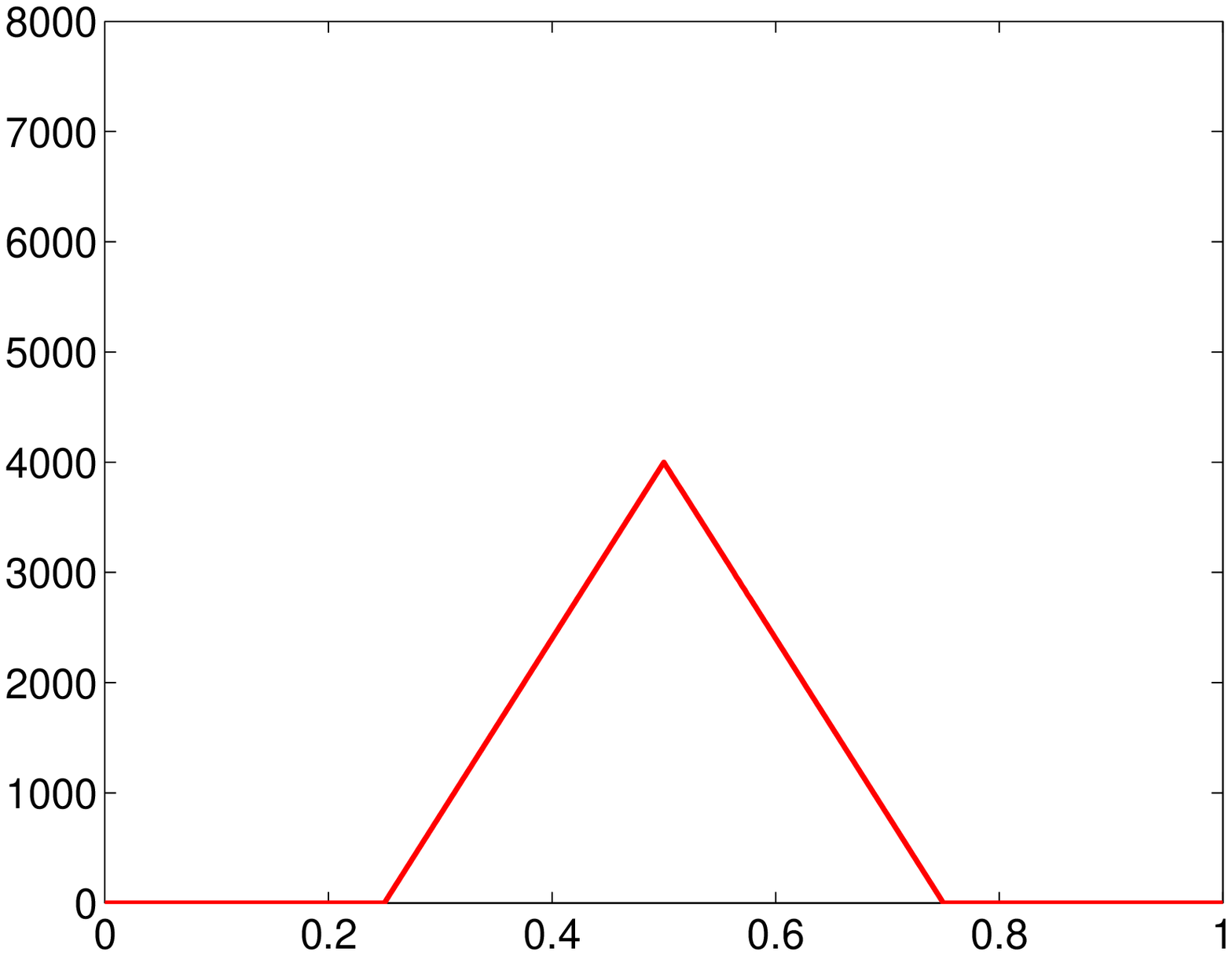}&  \hspace{-24pt}
 \includegraphics[height=1.3in]{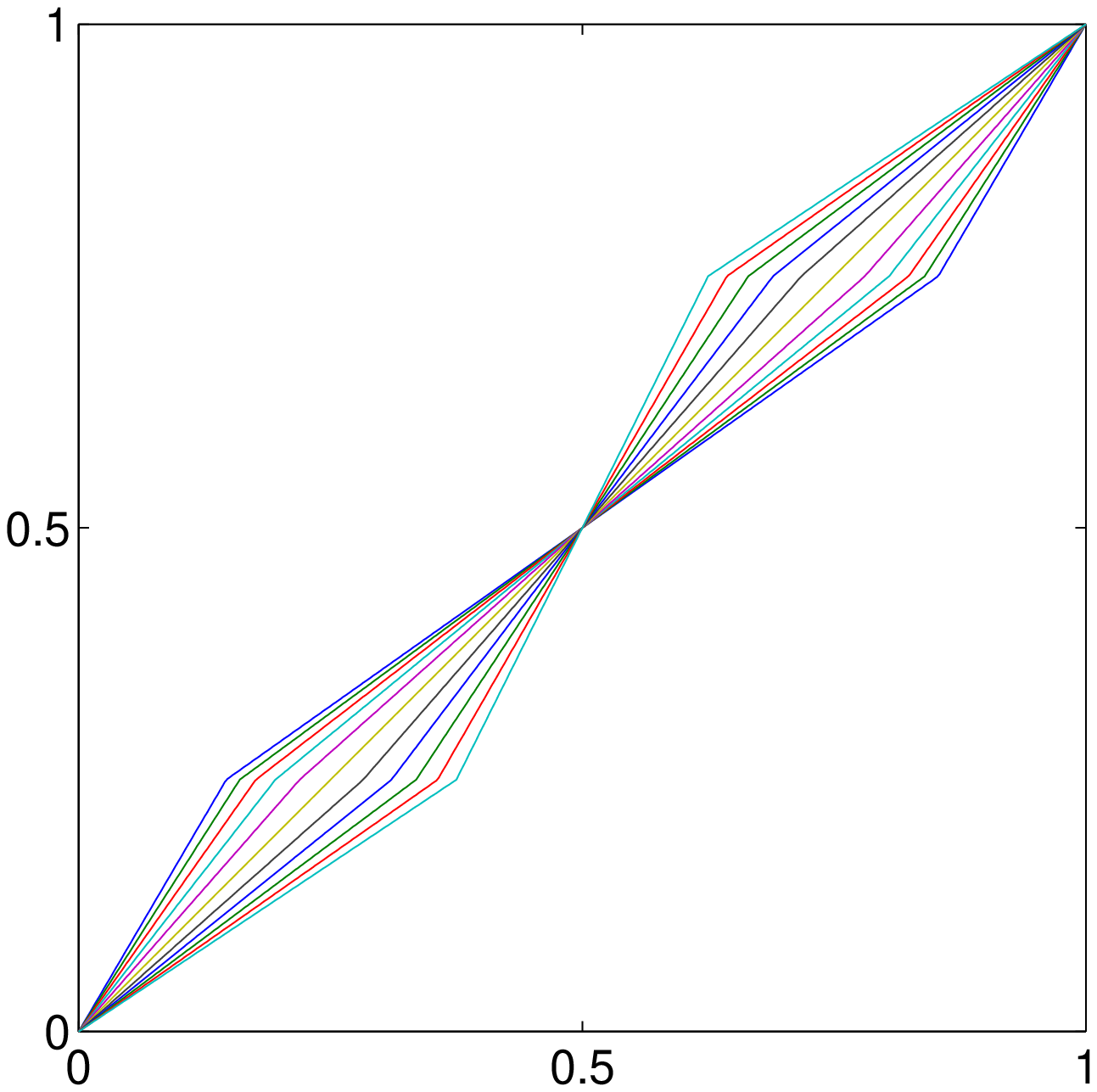}&  \hspace{-24pt}
  \includegraphics[height=1.3in]{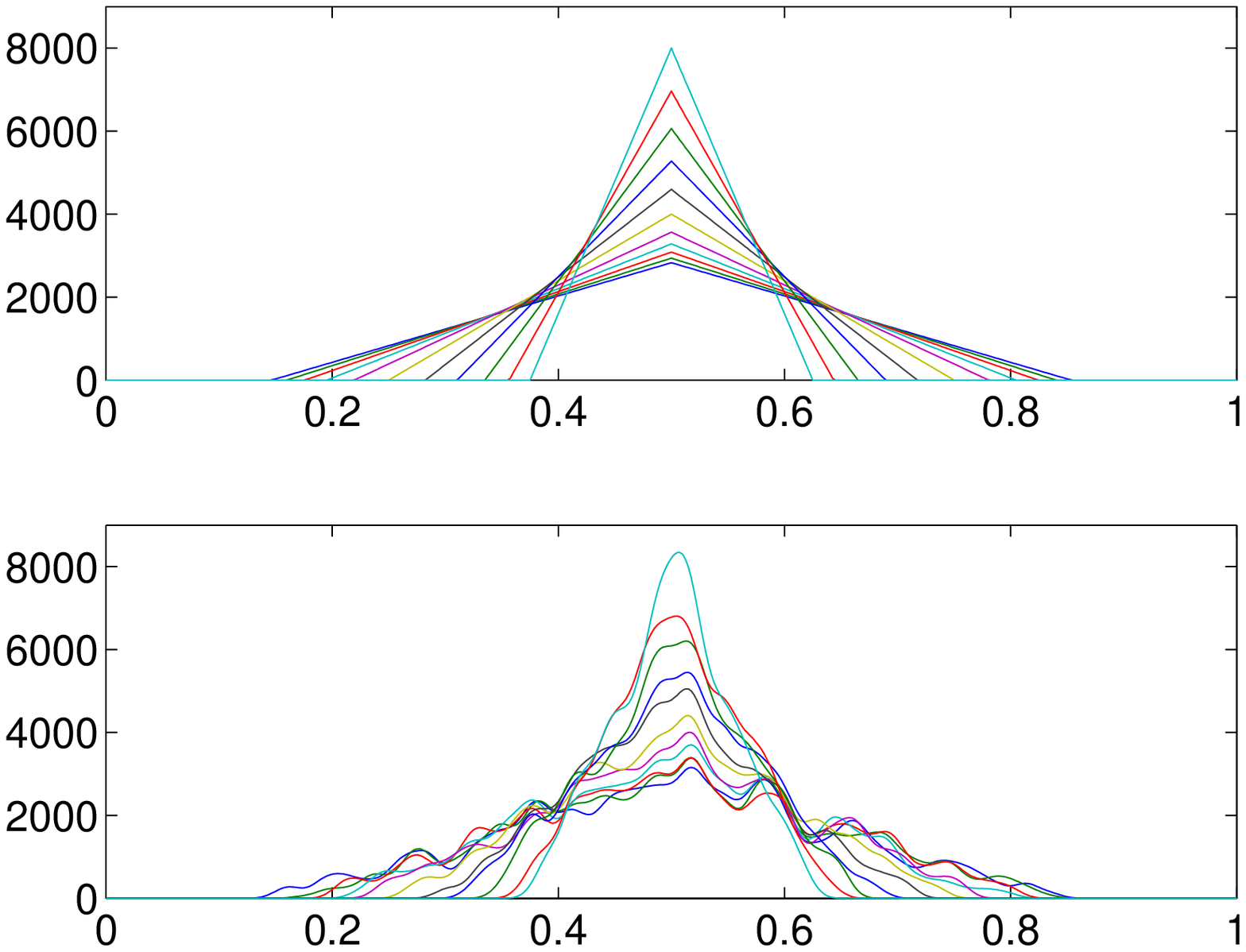}\\
  \hspace{-24pt} \textbf{D} &  \hspace{-24pt} \textbf{E} &  \hspace{-24pt} \textbf{F}\\  \hspace{-24pt}
  \includegraphics[height=1.3in]{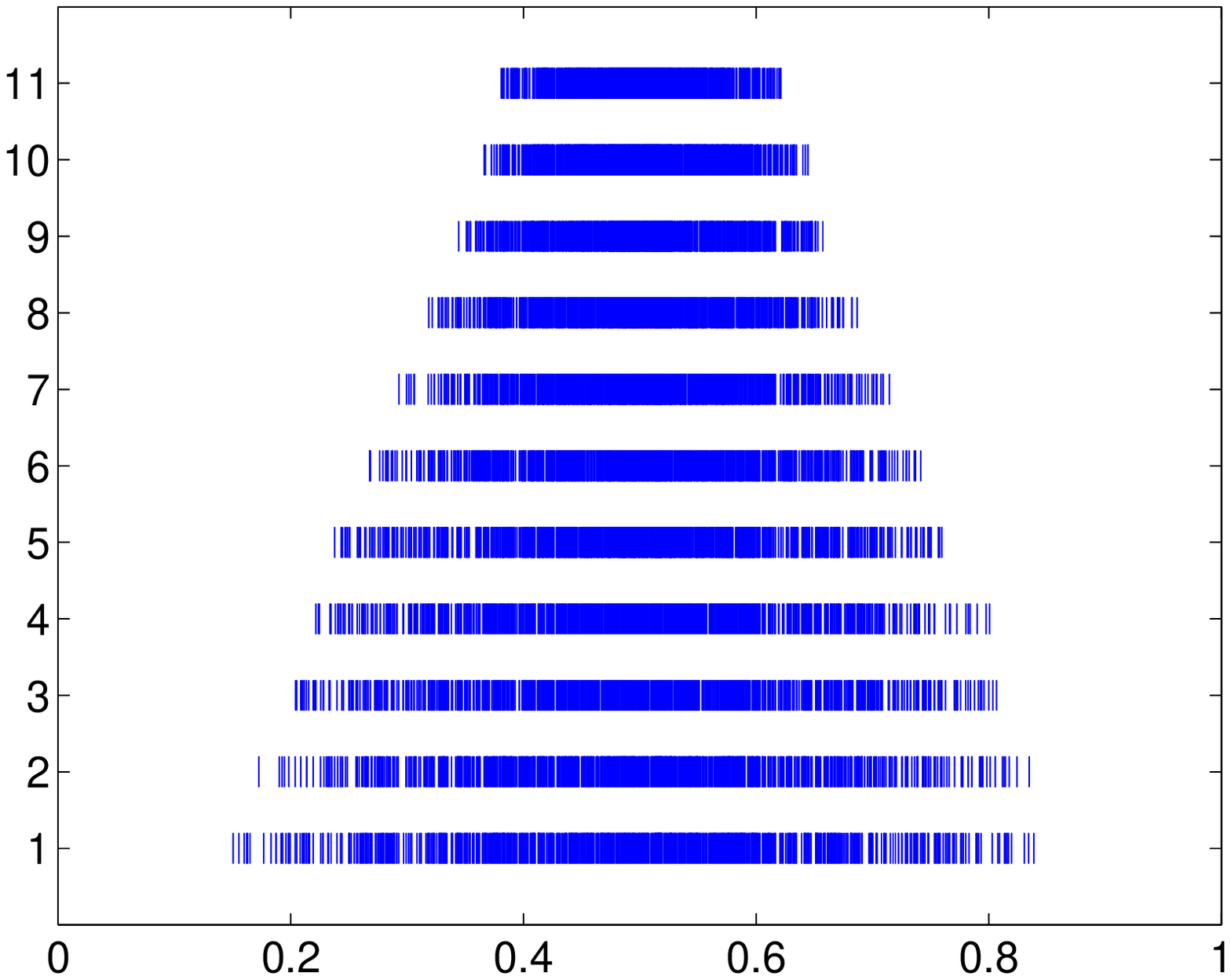}&  \hspace{-24pt}
  \includegraphics[height=1.3in]{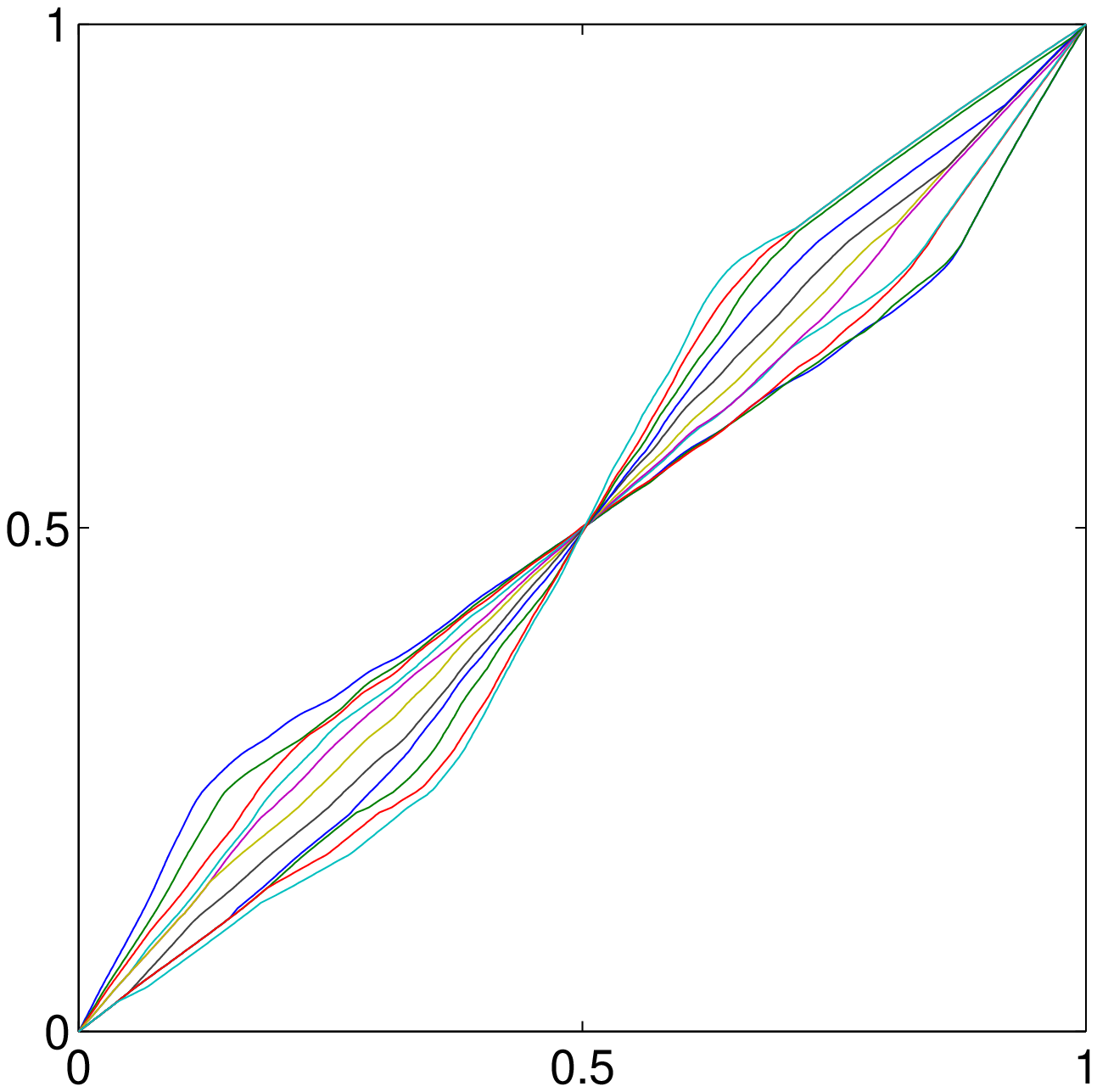} &  \hspace{-24pt}
\includegraphics[height=1.3in]{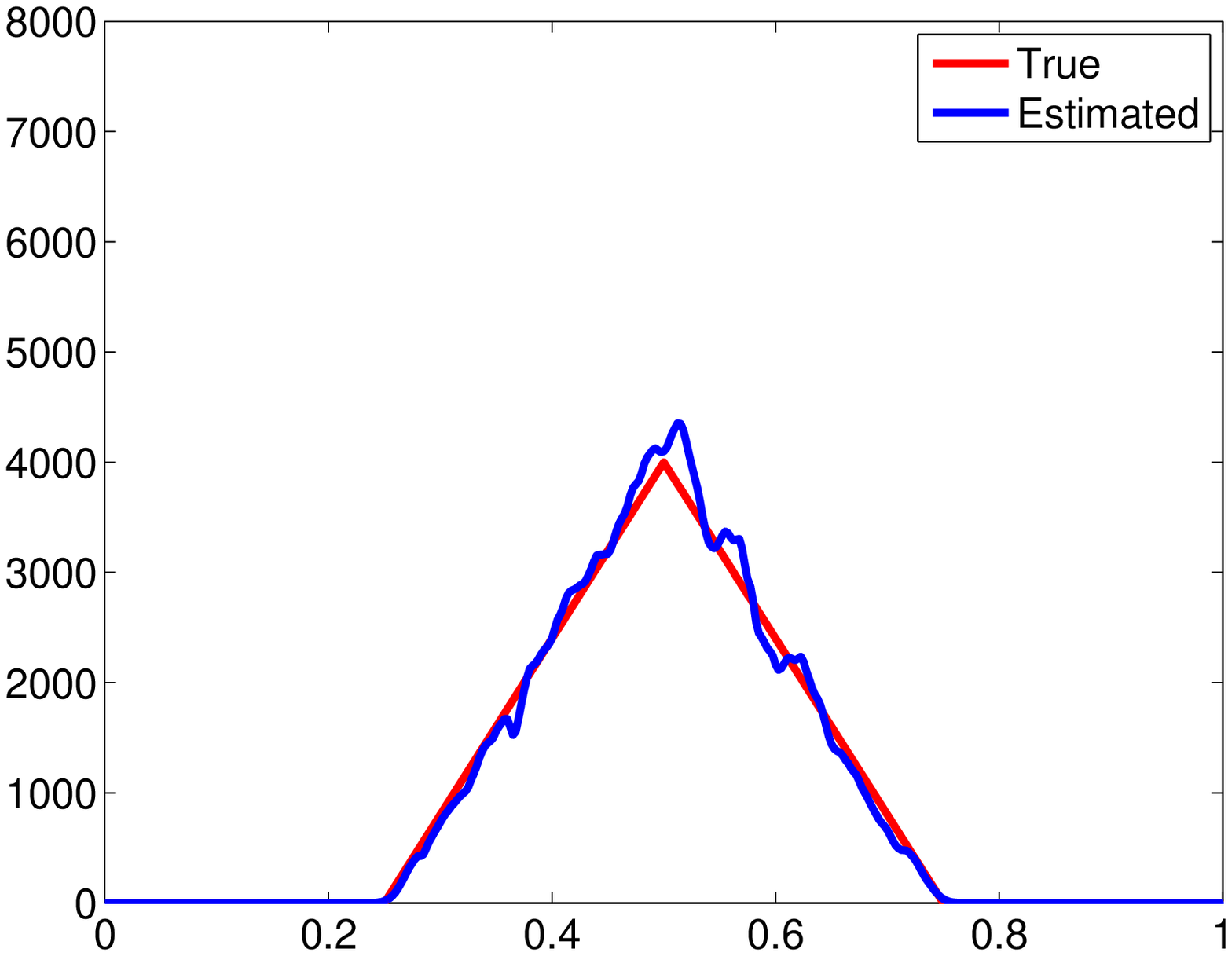} 
\end{tabular}
\caption{Non-Negative Intensity Estimation. \textbf{A.} True intensity function. \textbf{B.} 11 warping functions. \textbf{C.} (top panel) True warped intensity functions and (bottom panel) estimated intensities with modified kernel method.  \textbf{D.}  11 simulated processes with respect to the warped intensities. \textbf{E.} Estimated warping functions. \textbf{F.} Estimated (blue) and true (red) intensity functions} 
\label{fig:trisim}
\end{center}
\end{figure}

We then generate 11 warping functions $\{\gamma_i\}_{i = 1}^{11}$ in the following two steps:  At first, we define $\tilde{\gamma}_i \in \Gamma$ on [0, 1] as:
\begin{equation}
\tilde{\gamma}_i(t)=\dfrac{\mathrm{sign}(2t-1)|2t-1|^{e_i}+1}{2}
\end{equation}
where 
\[e_i = \left\{
  \begin{array}{lr}
    \dfrac{1}{2-0.2(i-1)} &  i=1,\dots,6\\
    0.2(i-6)+1 &  i=7,\dots,11.
  \end{array}
\right.
\]
Then, each $\gamma_i(t)$ is defined by linearizing $\tilde{\gamma}_i(t)$ at the value points $t=[0, 0.25, 0.5$, $0.75,1]$.  These warping functions are shown in Fig. \ref{fig:trisim}B.  The warped intensity functions, $\lambda_i(t) = 
\lambda (\gamma_i(t))\dot \gamma(t)$, are shown in the top panel of Fig. \ref{fig:trisim}C.   We then simulate 11 independent Poisson processes using these 11 intensity functions, respectively, and the results are shown in Fig. \ref{fig:trisim}D.  We can see that these realizations clearly display the warped intensity functions along the time axis.   
Given these noisy Poisson process observations, we aim to reconstruct the underlying intensity function $\lambda(t)$.  

To estimate $\lambda(t)$, we first estimate the warped intensity functions using modified kernel method on the 11 observed realizations.  We fitted a truncated Gaussian kernel with bandwidth $h=0.01$ to estimate the intensities.  The result is shown in the lower panel of Fig. \ref{fig:trisim}C.  Comparing with the true intensities in the corresponding upper panel, we can see the kernel method provides a reasonable estimation.  In spite of the phase shift along the time axis, the kernel method estimates the flat subregions in the underlying intensity appropriately.

Once the individual intensities are estimated, we then compute their Karcher mean to get the the warping functions with formula in Eqn. \ref{eq:opt}. These warping functions were then used to estimate of the underlying intensity function for the process and the result is shown in Fig. \ref{fig:trisim}F.  Comparing the result with the true intensity function, we find that the proposed method provides a very accurate reconstruction.

\subsection{Application in Spike Train Data} \label{subsec:spike1}

\begin{figure}[ht]
\begin{center}
\includegraphics[height=4.5in]{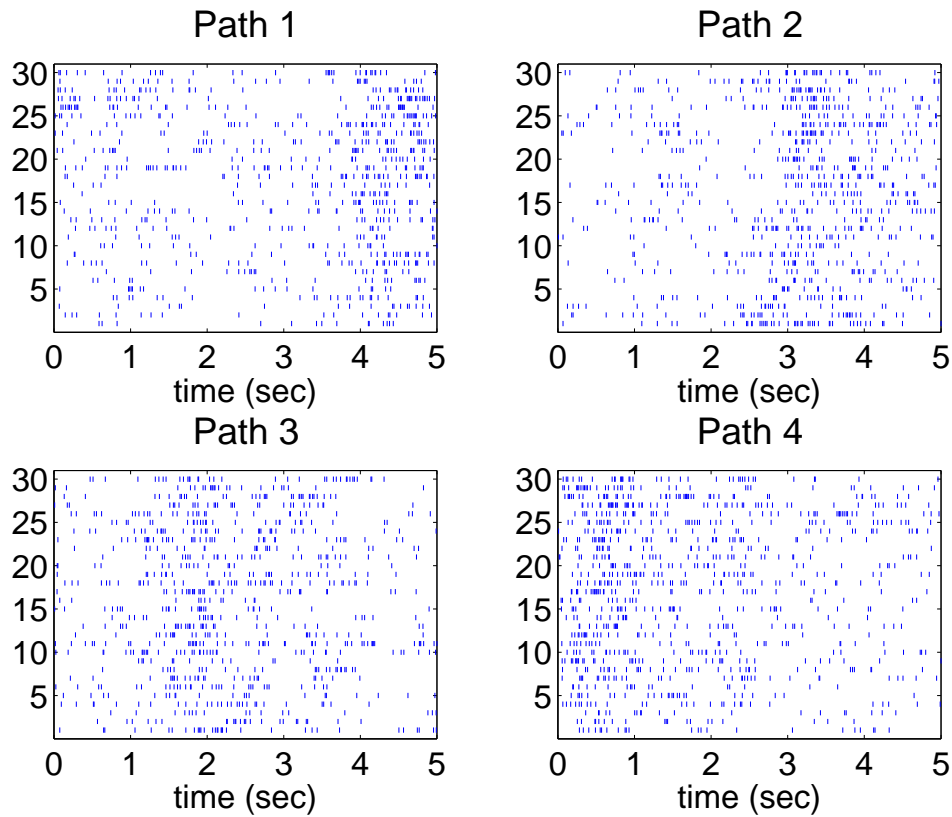} 
\caption{30 spike trains in each of the four movement paths. }
     \label{fig:data}
\end{center}
\end{figure}

\begin{figure}[ht]
\begin{center}
\includegraphics[height=4.5in]{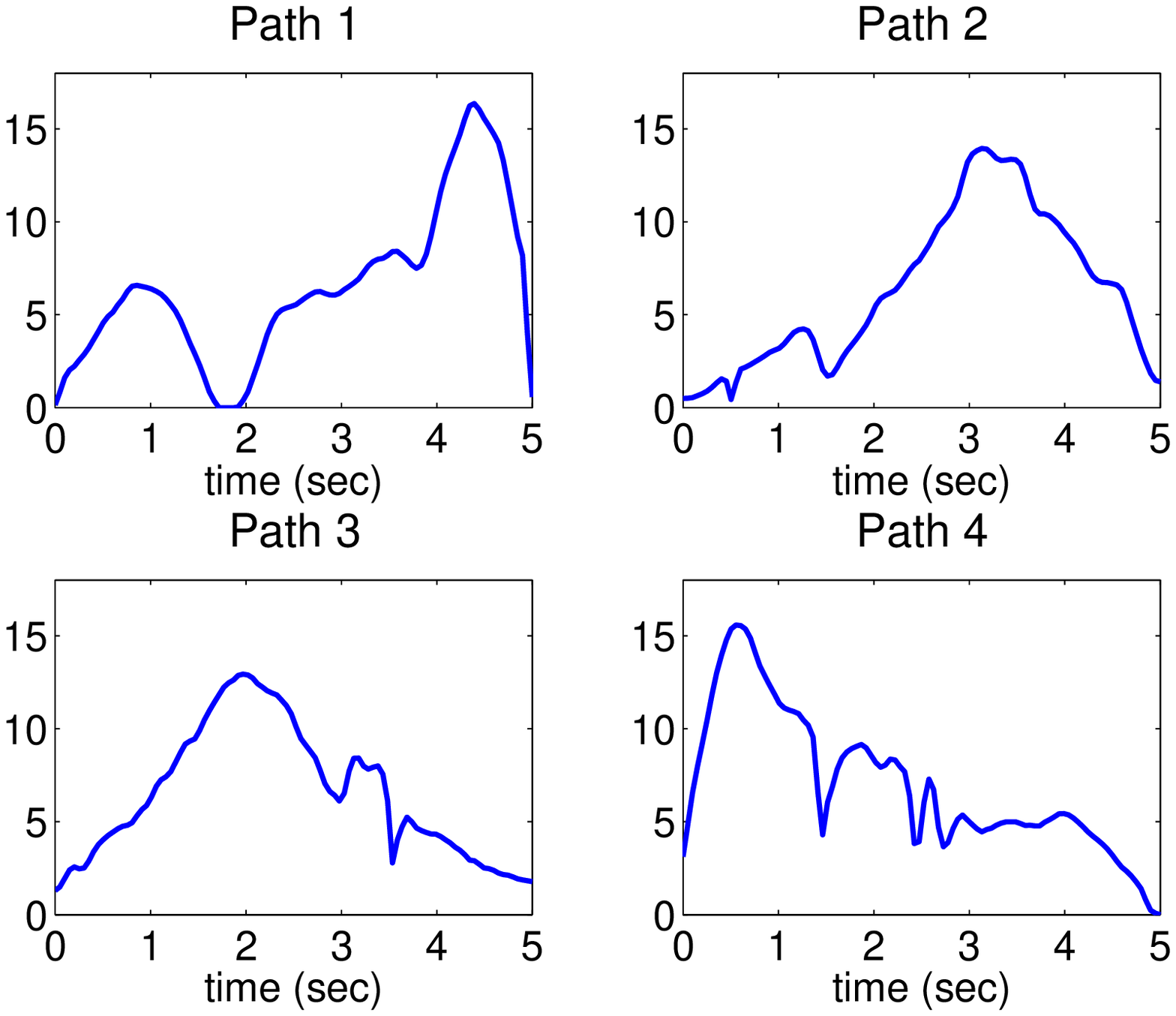}
\caption{Estimated intensity function in each path.}
 \label{fig:sp_mean}
\end{center}
\end{figure}

In this section the proposed intensity estimation method will be applied to a benchmark spike train dataset. This dataset was first used in a metric-based analysis of spike trains \cite{WuSrivastavaJCNS11}, and was also used as a common data set in a workshop on function registration, {\em CTW: Statistics of Time Warpings and Phase Variations} in Mathematical Bioscience Institute in 2012. It is publicly available from {\it http://mbi.osu.edu/2012/stwdescription.html} and is the same dataset used in Chapter 1. For completeness, a brief summary is given again. The spiking activity of one neuron in primary motor cortex was recorded in a juvenile female macaque monkey. In the experimental setting, a subject monkey was trained to perform a closed Squared-Path (SP) task by moving a cursor to targets via contralateral arm movements in the horizontal plane. Basically, the targets in the SP task are all fixed at the four corners of a square and the movement is stereotyped. In each trial the subject reached a sequence of 5 targets which were the four corners of the square with the first and last targets overlapping. Each sequence of 5 targets defined a path, and there were four different paths in the SP task (depending on the starting point).  In this experiment, 60 trials for each path were recorded, and the total number of trials was 240.

To fix a standardized time interval for all data, the spiking activity in each trial is normalized to 5 seconds. For the purpose of intensity estimation, a modified Gaussian kernel (width = 41.67$ms$) was adopted to estimate the underlying density of each of the point process spike trains. Thirty smoothed spike trains in each path are shown in Fig. \ref{fig:data}B. From these data, observe that the densities have a similar pattern within each class; for example, they have similar number of peaks and the locations of these peaks are only slightly different. However, the peak locations across different paths are significantly different. 
 
For the 60 trials in each path, the first 30 of them were chosen as the training data and the other 30 as the test data. The proposed intensity estimation method is tested here to decode neural signals with respect to different movement paths. In general, there are two types of decoding methods: i) classification based on pairwise distance between training and test data, and ii) classification using distance from test data to the Karcher mean in the training data. Note that the pairwise method has a quadratic efficiency (Cost is $O(N^2)$, where $N$ is the number of spike trains in training and testing set), but distance-to-the-mean is in the linear order \cite{WuSrivastavaJCNS11}. In this chapter, the decoding result is reported using the efficient mean-based method.   

Once an estimate for the density of each of the spike trains was obtained, the Karcher mean for each path was calculated using Algorithm 1 with one minor change to overcome numerical issues. In step 2, instead of directly using the CDF and inverse CDF of the two densities, the individual warping functions are found using Dynamic Programming \cite{srivastava-etal-Fisher-Rao-CVPR:2007}. The penalty coefficient used in the Dynamic Programming was 0.01, although the results are robust to the choice of this penalty coefficient. The computed Karcher means in each path are shown in Fig. \ref{fig:sp_mean}.  

Comparing with the original spike trains, all of the mean spike trains appropriately represent the firing patterns in the corresponding movement. For example, the spiking frequency is relatively higher when the hand moves upward, which is apparent in all four means. For the 120 test trains, each train is labeled by the shortest distance over the distances to the four means in the training set.  This computation is apparently more efficient (only $120 \times 4 = 480$ distances need to be computed). It is found that the classification
accuracy using the proposed estimation method is 82.5\%(99/120) whereas the classification accuracies using the naive cross-sectional method and the Fisher-Rao registration method are 77.5\%(93/120) and 55.0\%(66/120), respectively. This result shows the proposed method can better differentiate neural signals with respect to different movement behaviors. The lower accuracy in the naive method indicates that the proposed method improves classification results.   

\section{Discussion}\label{sec:diss}
Intensity estimation has been a classical problem in Poisson process methods.  The problem is significantly challenging if the observed data are corrupted with compositional noise, i.e. there is time warping noise in each realization.  In the paper, we have proposed a novel alignment-based algorithm for positive intensity estimation.  The method is based on a key fact that the intensity function is area-preserved with respect to compositional noise. Such a property implies that the time warping is only encoded in the normalized intensity, or density, function.  Based on this finding, we decompose the estimation of intensity by the product of estimated total intensity and estimated density.  Our investigation on asymptotics shows that the proposed estimation algorithm provides a consistent estimator for the underlying density.  We further extend the method to all nonnegative intensity functions, and provide simulation examples to illustrate the success of the estimation algorithms. 

While results from this method show promising improvements over previous methods, it is important to note that the method is dependent upon the kernel density estimates of the observed processes. In general, kernel density estimates are highly dependent upon the chosen bandwidth $h$ \cite{ramsay-silverman-2005, ferraty-vieu-2006}. In this paper, we have used a simple plug-in method to determine an appropriate bandwidth. In future work, we will consider the development of an algorithm that can automatically choose the optimal bandwidth for the modified kernel density estimator. Additionally, future work will examine the asymptotic variability of this estimator and an extension to general Cox processes for conditional intensity estimation. 

\section*{Appendix}
\subsection*{A. Proof on proper metric $d_{ext}$} 

\begin{proof} We prove that $d_{ext}$ is a proper metric by verifying three properties: 
\begin{enumerate}
\item (Positive Definiteness) It is apparent that $d_{ext}(f_1, f_2) \ge 0$.  By Theorem 1, there exists $\gamma_{12}$, such that $f_1 = (f_2; \gamma_{12})$.  Therefore, $d_{ext}(f_1,f_2)=0
 \Leftrightarrow \Vert 1 - \sqrt{\dot{\gamma}_{12}} \Vert=0 \Leftrightarrow \gamma_{12}(t)=\gamma_{id}$. Hence, $f_1=f_2$.
\item (Symmetry)
$\Vert 1-\sqrt{\dot{\gamma}_{21}}\Vert^2=\Vert 1-\sqrt{\dot{\gamma}_{12}^{-1}}\Vert^2 = \int_0^1 \left( 1-\sqrt{\dot{\gamma}_{12}^{-1}(s)}\right)^2 ds 
= \int_0^1 \left(1-\dfrac{1}{\sqrt{\dot{\gamma}_{12}(t)}}\right)^2\dot{\gamma}_{12}(t)dt
= \Vert 1-\sqrt{\dot{\gamma}_{12}(t)}\Vert^2.$
Therefore, $d_{ext}(f_1,f_2)=d_{ext}(f_2,f_1)$.

\item (Triangle Inequality) 
Let $f_2=f_1\left( \gamma_{12}(t)\right)$, $f_3=f_2\left(\gamma_{23}(t)\right)$, $\gamma_{13}=\gamma_{12}\circ\gamma_{23}$. Then, 
$d_{ext}(f_1, f_3) = \Vert 1-\sqrt{\dot{\gamma}_{13}} \Vert 
= \Vert 1-\sqrt{\left(\dot{\gamma}_{12}\circ \gamma_{23}\right)\dot{\gamma}_{23}} \Vert 
\le \Vert 1-\sqrt{\dot{\gamma}_{23}}\Vert+\Vert\sqrt{\dot{\gamma}_{23}}-\sqrt{\dot{\gamma}_{13}}\Vert 
= \Vert 1-\sqrt{\dot{\gamma}_{23}}\Vert+\Vert 1 - \sqrt{\dot{\gamma}_{12}}\Vert$.   Note that 
$\Vert\sqrt{\dot{\gamma}_{23}}-\sqrt{\dot{\gamma}_{13}}\Vert 
=  \Vert \left(1,\gamma_{23}\right)-\left(1,\gamma_{13}\right)\Vert
= \Vert \left(1,\gamma_{23}\right)-\left(1,\gamma_{12}\circ\gamma_{23}\right)\Vert
= \Vert 1 - \sqrt{\dot{\gamma}_{12}}\Vert$  (by isometry)
Thus, $d_{ext}(f_1, f_3) \le d_{ext}(f_1, f_2) + d_{ext}(f_2, f_3)$.
\end{enumerate}
\end{proof}

\subsection*{B. Proof on the consistency of $\hat f$} 

\begin{lemma}
\label{lem:kernel}
Let $g$ be a probability density function on [0, 1].  $\{X_i\}_{i=1}^n$ are a set of i.i.d. random variables with density $g$.  If $\hat g_n$ is a modified kernel estimate with optimal bandwidth given in Algorithm 2, then
$$ \int_0^1 |\hat g_n(t) - g(t)|dt \xrightarrow{a.s.} 0 \ (\mbox{when } n \rightarrow \infty) $$ 
\end{lemma}

\begin{proof}
Let $\tilde g_n(t) = \frac{1}{nh_n} \sum_{i=1}^n K(\frac{t-X_i}{h_n})$ be the classical kernel estimator with kernel function $K$ and optimal bandwidth $h_n$ (i.e. $h_n \rightarrow 0$ and $nh_n \rightarrow \infty$).  Then,  we can obtain from Equation 3.84 of \cite{silverman86} that $\int_0^1 |\tilde g_n(t) - g(t)|dt \xrightarrow{a.s.} 0$.  

As $K(t) = 0$ when $|t| < 1$, we have 
\begin{eqnarray}
& & \int_0^1 |\hat g_n(t) - \tilde g_n(t)|dt \nonumber \\
&=& \int_0^{h_n} |\hat g_n(t) - \tilde g_n(t)|dt + \int_{1-h_n}^1 |\hat g_n(t) - \tilde g_n(t)|dt  
 + \int_{h_n}^{1-h_n} |\frac{1}{n+1}\tilde g_n(t) + \frac{1}{n+1}|dt  \nonumber \\
&\le& \int_0^{h_n} \hat g_n(t)dt + \int_0^{h_n} \tilde g_n(t) dt + \int_{1-h_n}^1 \hat g_n(t) dt + \int_{1-h_n}^1 \tilde g_n(t) dt + \frac{2}{n+1}.  \label{eq:kernel}
\end{eqnarray}
Here we will show that the first term goes to 0 (a.s.).  Indeed, 
\begin{eqnarray*}
\int_0^{h_n} \hat g_n(t)dt &=& \int_0^{h_n} \frac{1}{nh_n} \sum_{i=1}^n K(\frac{t-X_i}{h_n}) dt  
=  \int_0^{h_n} \frac{1}{nh_n} \sum_{X_i \le 2h_n} K(\frac{t-X_i}{h_n}) dt  \\
& \le & \int_0^{1} \frac{1}{nh_n} \sum_{X_i \le 2h_n} K(\frac{t-X_i}{h_n}) dt  
= \frac{1}{n} \sum_{X_i \le 2h_n} 1  
= \frac{1}{n} \sum_{i=0}^1 \mathbf{1}_{\{{X_i \le 2h_n}\}}
\end{eqnarray*}
where $\mathbf{1_{\{\cdot\}}}$ is the indicator function.  By the Strong Law of Large Numbers on triangular arrays  \cite{wangwu11}, 
$$\frac{1}{n} \sum_{i=0}^n (\mathbf{1}_{\{{X_i \le 2h_n}\}} - E \mathbf{1}_{\{{X_i \le 2h_n}\}}) \rightarrow 0. (a.s.)$$ 
As  $ E \mathbf{1}_{\{{X_i \le 2h_n}\}} =  \int_0^{2h_n} f(t)dt \rightarrow 0,$ we have
$\int_0^{h_n} \hat g_n(t)dt \xrightarrow{a.s.} 0$.  The convergence to 0 for the second to fourth terms on the RHS of Eqn. \ref{eq:kernel} can be similarly proven, and therefore  $\int_0^1 |\hat g_n(t) - \tilde g_n(t)|dt \xrightarrow{a.s.} 0. $  Finally, we have 
$$\int_0^1 |\hat g_n(t) - g(t)|dt \le \int_0^1 |\hat g_n(t) - \tilde g_n(t)|dt + \int_0^1 |\tilde g_n(t) - g(t)|dt \xrightarrow{a.s} 0.$$    
\end{proof}

\begin{lemma}
\label{lem:inverse}
Let $G$ and $\hat G_n$ denote the cumulative distribution functions of $g$ and $\hat g_n$ in Lemma \ref{lem:kernel}, respectively.  Assume the density $g$ is continuous and for any $t \in [0,1]$, $0 < m \le g(t) \le M <\infty$ (Condition 2 in Sec. \ref{sec:theory}). 
If $G$ and $\hat G_n$ are invertible and the inverse functions are differentiable, then 
$$ \int_0^1 (\sqrt{{\dot {\hat G}}_n^{-1} (t)} - \sqrt{{\dot G}^{-1}(t)})^2dt \xrightarrow{a.s.} 0 \ (\mbox{when } n \rightarrow \infty) $$ 
\end{lemma}
\begin{proof}
To simplify notation, we let $F = G^{-1}$, $\hat F_n = \hat G_n^{-1}$, $f = \dot F = {\dot G}^{-1}$, and $\hat f_n =  {\dot {\hat F}}_n = {\dot {\hat G}}_n^{-1} $. 
For any $t \in [0, 1], |\hat G_n (t) - G(t)| \le \int_0^1 |\hat g_n(t) - g(t)| dt \xrightarrow{a.s.} 0$ (by Lemma 3).   That is, $\hat G_n \rightrightarrows G$ (uniform convergence) almost surely.  By the theory on convergence of inverse functions \cite{wangwu11}, we also got that  $\hat F_n \rightrightarrows F \ (a.s.)$.  

By definition, $G(F(t)) = t$ and $ \hat G_n (\hat F_n (t)) = t$.  Using the chain rule, we have $g(F(t)) f(t) = 1$ and $\hat g_n(\hat F_n(t)) \hat f_n(t) = 1$.  Therefore, 
\begin{eqnarray*}
 & &\int_0^1  |\hat f_n (t) - f(t)|dt  \\
 &=& \int_0^1 |\frac{1}{\hat g_n(\hat F_n(t))} - \frac{1}{g(F(t))} | dt \\
 &\le& \int_0^1 |\frac{1}{\hat g_n(\hat F_n(t))} - \frac{1}{g(\hat F_n(t))} | dt + \int_0^1 |\frac{1}{g(\hat F_n(t))} - \frac{1}{g(F(t))} | dt \\
\end{eqnarray*}
Here we will show that each integration in the right-hand side indeed converges to 0 (a.s.).  By Lemma \ref{lem:kernel},
\begin{eqnarray*}
& & \int_0^1 |\frac{1}{\hat g_n(\hat F_n(t))} - \frac{1}{g(\hat F_n(t))} | dt \\
&=& \int_0^1 |\frac{1}{\hat g_n(s)} - \frac{1}{g(s)} | \hat g_n(s) ds \   \mbox{ (by change of variable)}  \\
&=& \int_0^1 \frac{1}{g(s)}  |{\hat g_n(s)} - {g(s)} | ds \le \frac{1}{m} \int_0^1 |{\hat g_n(s)} - {g(s)} | ds \xrightarrow{a.s.} 0 
\end{eqnarray*}
By assumption, $g$ is continuous and positively bounded. Hence, $1/g$ is also continuous.  This continuity is uniform because the domain $[0,1]$ is compact.  That is, for any $\epsilon > 0$, there exists $\delta >0$, such that for all $a,b \in [0,1]$ with $|a-b| < \delta, |1/g(a) - 1/g(b)| < \epsilon$.  We have shown that $\hat F_n \rightrightarrows F \ (a.s.)$.  Hence, with probability 1,
there exists an integer $N$ such that for any $n > N$ and $t \in [0, 1]$, we have $|\hat F_n(t) - F(t)| < \delta.$ 
$\int_0^1 |\frac{1}{g(\hat F_n(t))} - \frac{1}{g(F(t))} | dt  \le \int_0^1 \epsilon dt = \epsilon. $  Therefore, we have shown that 
\begin{eqnarray*}
\int_0^1 |\frac{1}{g(\hat F_n(t))} - \frac{1}{g(F(t))} | dt   \xrightarrow{a.s.} 0.
\end{eqnarray*}
Finally, based on the simple inequality $(\sqrt{a} - \sqrt{b})^2 \le |a-b|$, we have
$$  \int_0^1  (\sqrt{\hat f_n (t)} - \sqrt{f(t)})^2dt   \le  \int_0^1  |\hat f_n (t) - f(t)|dt  \xrightarrow{a.s.} 0. $$
\end{proof}

\begin{lemma}
\label{lem:karcher}
Let $\{\gamma_i\}$ be a sequence of warping functions that satisfy Condition 3 in Sec. \ref{sec:theory}, and $\bar \gamma$ be the Karcher mean of $\{\gamma_i^{-1}\}$.  Then $\bar \gamma$ converges to $\gamma_{id}$ almost surely.  That is, 
$$ || (1, \bar \gamma) - 1|| \xrightarrow{a.s.} 0 \ \ \  (\mbox{when } n \rightarrow \infty)$$
\end{lemma}
\begin{proof}
By assumption, $E(\sqrt {\dot \gamma_i^{-1}(t)}) \equiv \beta > 0, i = 1, \cdots, n.$  Let $S_n = n \bar \gamma = \sum_{i=1}^n \sqrt {\dot \gamma_i^{-1}}$. As $\{\sqrt {\dot \gamma_i^{-1}}\}$ are i.i.d., 
\begin{eqnarray*}
& &E\left(\left\Vert S_n -n \beta\right\Vert ^4\right) = E\left(\left\Vert\sum_{i=1}^n \left(\sqrt {\dot \gamma_i^{-1}} - \beta\right)\right\Vert^4\right) \\
&=& nE\left(\left\Vert\sqrt {\dot \gamma_1^{-1}} - \beta\right\Vert^4\right) + n(n-1) \left(E\left(\left\Vert\sqrt {\dot \gamma_1^{-1}} - \beta\right\Vert^2\right)\right)^2 \\
& & + 2n(n-1)  \left(E\left(\int_0^1\left(\sqrt {\dot \gamma_1^{-1}(t)} - \beta\right) \left(\sqrt{\dot \gamma_2^{-1}(t)} - \beta\right)dt\right)^2\right)
\end{eqnarray*}
As $||\sqrt {\dot \gamma_1^{-1}}|| = 1$, there exist positive constants $C$ and $N$, such that $E(|| S_n -n \beta||^4) < Cn$ when $n > N$.  

Using the generalized Chebyshev inequality, for any $\epsilon > 0$ and $n > N$, 
$$ P\left(\left\Vert\frac{S_n - n\beta}{n}\right\Vert > \epsilon\right) \le \frac{1}{(n\epsilon)^4} E(|| S_n -n \beta||^4) \le \frac{C}{\epsilon^4 n^2}. $$
This indicates that $\sum_{n=1}^\infty P(||S_n-n\beta|| \ge n\epsilon) < \infty$.  By the Borel-Cantelli lemma,
$P(||S_n-n\beta|| \ge n\epsilon \ \ i.o.) = 0$. Therefore, 
$ || \frac{1}{n}\sum_{i=1}^n \sqrt {\dot \gamma_i^{-1}} - \beta || \xrightarrow{a.s.} 0. $  
Finally, we have 
\begin{eqnarray*}
|| (1, \bar \gamma) - 1||  &=& || \sqrt {\dot {\bar \gamma}}  - 1 ||  = \left\Vert \frac{\frac{1}{n}\sum_{i=1}^n \sqrt {\dot {\gamma}_i^{-1}}} {\frac{1}{n}||\sum_{i=1}^n \sqrt {\dot {\gamma}_i^{-1}} ||} - 1 \right\Vert 
\xrightarrow{a.s.} 0. 
\end{eqnarray*}

\end{proof}

\begin{lemma}
\label{lem:poisson}
Assume $Y_m$ is a random variable following a Poisson distribution with mean $\Lambda_m$.  If $\lambda_m \ge \alpha \log(m), \alpha > 1$ for sufficiently large $m$ (Condition 4 in Sec. \ref{sec:theory}), then $Y_m \rightarrow \infty \ (a.s.)$ when $m \rightarrow \infty$. 
\end{lemma}

\begin{proof}
Based on the Poisson density formula, for any $K = 1, 2, \cdots,$ $P(Y_m \le K) = e^{-\lambda_m}  \sum_{k=0}^K \frac{\lambda_m^k}{k!}$.  By assumption, $\lambda_m \ge \alpha \log(m), \alpha > 1$ for sufficiently large $m$.  It is apparent that when $m$ is sufficiently large, $e^{-\frac{\alpha/2}{1+\alpha}\lambda_m} \sum_{k=0}^K \frac{\lambda_m^k}{k!} < 1.$  Hence, 
\begin{eqnarray*}
m^{1+\alpha/2} P(Y_m\le K) 
&=& m^{1+\alpha/2} e^{-\lambda_m}  \sum_{k=0}^K \frac{\lambda_m^k}{k!} \\
&=&  e^{-\frac{1+\alpha/2}{1+\alpha}\lambda_m +(1+\alpha/2) \log m}  \left(e^{-\frac{\alpha/2}{1+\alpha}\lambda_m}\sum_{k=0}^K \frac{\lambda_m^k}{k!}\right)  \\
&\le& 1 \cdot  1 = 1. 
 \end{eqnarray*}
Consequently, for sufficiently large $m$, $P(Y_m\le K) \le  \frac{1}{m^{1+\alpha/2}}$.   Hence, $\sum_{m=1}^\infty P(Y_m\le K) < \infty$.  By the Borel-Cantelli lemma, 
$$ P(\limsup \{Y_m \le K\}) = P(Y_m \le K \ \ i.o.) = 0.$$ 
Equivalently, we have $P(Y_m > K \ \ \mbox{eventually}) = 1,$ for any $K = 1, 2, ...$.  Therefore, 
$$ \lim_{m \rightarrow \infty} Y_m = \infty. \ \ (a.s.)$$
\end{proof}

\bibliographystyle{apa}    
\bibliography{references}  

\end{document}